%% file: main.tex
\documentclass[a4paper,onecolumn,11pt,published]{quantumarticle}
\pdfoutput=1

\usepackage[utf8]{inputenc}
\usepackage{amsmath}
\usepackage{amssymb}
\usepackage{bbold}
\usepackage{comment}
\usepackage{graphicx}
\usepackage{scalerel}
\usepackage{braket}
\usepackage{leftidx}
\usepackage[title]{appendix}
\usepackage{tikz}
\usepackage{tikzit}
\usepackage{tikz-cd}
\usepackage{cite}
\usepackage{geometry}
\usepackage{mathtools}
\usepackage{ stmaryrd }
\usepackage{circuitikz}
\usepackage{textcomp}
\usepackage{accents}
\usepackage{comment}
\usepackage{ mathrsfs }
\usepackage{authblk}
\usepackage{breakurl,amsmath,amsthm,amssymb,xspace,xypic,enumerate,color}
\usepackage{float}

\tikzcdset{scale cd/.style={every label/.append style={scale=#1},
    cells={nodes={scale=#1}}}}

\input{tikzstyles}

\input{stylesnew5.tikzstyles}
\usetikzlibrary{calc, positioning, shapes.geometric}
\usetikzlibrary{
	arrows,
	shapes,
	decorations,
	intersections,
	backgrounds,
	positioning,
	circuits.ee.IEC
	}

\newcommand{\tikzfigscale}[2]{\scalebox{#1}{\input{#2.tikz}}}

\def\be{\begin{equation}}
\def\ee{\end{equation}}
\def\ba{\begin{align}}
\def\ea{\end{align}}

\usepackage{bbm}

\newcommand{\lalg}[1]{\mathfrak{#1}}  

\renewcommand{\sl}{\lalg{sl}}

\newtheorem{definition}{Definition}
\newtheorem{theorem}{Theorem}
\newtheorem*{theorem*}{Theorem}
\newtheorem{corollary}{Corollary}
\newtheorem{lemma}{Lemma}

\newtheorem{example}{Example}

\tikzstyle{every picture}=[baseline=-0.25em,shorten <=-0.1pt]
\tikzstyle{dotpic}=[scale=0.5]
\tikzstyle{braceedge}=[decorate,decoration={brace,amplitude=1mm,raise=-1mm}]
\tikzstyle{dot}=[inner sep=0.7mm,minimum width=0pt,minimum height=0pt,fill=black,draw=black,shape=circle]
\tikzstyle{small dot}=[inner sep=0.1mm,minimum width=0pt,minimum height=0pt,fill=black,draw=black,shape=circle]
\tikzstyle{black dot}=[dot]
\tikzstyle{white dot}=[dot,fill=white]
\tikzstyle{gray dot}=[dot,fill=gray!40!white]
\tikzstyle{alt white dot}=[white dot,label={[xshift=3mm,yshift=-0.05mm,font=\tiny]left:$*$}]
\tikzstyle{alt gray dot}=[gray dot,label={[xshift=3mm,yshift=-0.05mm,font=\tiny]left:$*$}]
\tikzstyle{white norm}=[rectangle,fill=white,draw=black,minimum height=2mm,minimum width=2mm,inner sep=0pt,font=\small]
\tikzstyle{gray norm}=[white norm,fill=gray!40!white]
\tikzstyle{square box}=[rectangle,fill=white,draw=black,minimum height=5mm,minimum width=5mm,font=\small]
\tikzstyle{square gray box}=[rectangle,fill=gray!30,draw=black,minimum height=6mm,minimum width=6mm]
\tikzstyle{diredge}=[->]
\tikzstyle{rdiredge}=[<-]
\tikzstyle{dashed edge}=[dashed]
\tikzstyle{cross}=[preaction={draw=white, -, line width=3pt}]

\newcommand{\dotdualmult}[1]{%
\!\begin{tikzpicture}[dotpic]
    \node [style=white dot] (0) at (0, 0.3) {};
    \node [style=none] (1) at (-0.5, -0.4) {};
    \node [style=none] (2) at (0.5, -0.4) {};
    \node [style=none] (3) at (0, 0.8) {};
    \draw [style=diredge] (3.center) to (0);
    \draw [style=diredge, in=15, out=-30, looseness=1.50] (0) to (1.center);
    \draw [style=diredge, in=165, out=-150, looseness=1.50] (0) to (2.center);
\end{tikzpicture}\!}

\newcommand{\dotconorm}[1]{%
\,\begin{tikzpicture}[dotpic,yshift=0.4mm]
    \node [style=none] (0) at (0, -0.4) {};
    \node [style=white norm] (1) at (0, 0.1) {};
    \node [style=none] (2) at (0, 0.5) {};
    \draw [style=diredge] (1) to (0.center);
    \draw (2.center) to (1);
\end{tikzpicture}\,}

\newcommand{\astfootnote}[1]{
\let\oldthefootnote=\thefootnote
\setcounter{footnote}{0}
\renewcommand{\thefootnote}{\fnsymbol{footnote}}
\footnote{#1}
\let\thefootnote=\oldthefootnote
}

\title{Quantum Supermaps are Characterized by Locality}

\author{Matt Wilson}
\email{matthew.wilson@centralesupelec.fr}
\affiliation{Quantum Group, Department of Computer Science, University of Oxford}
\affiliation{HKU-Oxford Joint Laboratory for Quantum Information and Computation}
\affiliation{Programming Principles Logic and Verification Group,
University College London
London, UK}
\affiliation{Universit\'{e} Paris-Saclay, CNRS, ENS Paris-Saclay, Inria, CentraleSup\'{e}lec, Laboratoire M\'{e}thodes Formelles}

\author{Giulio Chiribella}
\email{giulio.chiribella@cs.ox.ac.uk}
\affiliation{Quantum Group, Department of Computer Science, University of Oxford}
\affiliation{HKU-Oxford Joint Laboratory for Quantum Information and Computation}
\affiliation{QICI Quantum Information and Computation Initiative, School of Computing and Data Science, The University of Hong Kong}
\affiliation{Perimeter Institute for Theoretical Physics, Waterloo, Ontario N2L 2Y5, Canada }

\author{Aleks Kissinger}
\email{aleks.kissinger@cs.ox.ac.uk}
\affiliation{Quantum Group, Department of Computer Science, University of Oxford}

\usepackage{hyperref}
\makeatletter
\def\HyPsd@expand@utfvii{}
\makeatother

\begin{document} \emergencystretch 3em

\maketitle

\begin{abstract}
We provide a new characterisation of quantum supermaps in terms of an axiom that refers only to sequential and parallel composition. Consequently, we generalize quantum supermaps to arbitrary monoidal categories and operational probabilistic theories. We do so by providing a simple definition of \textit{locally-applicable transformation} on a monoidal category. The definition can be rephrased in the language of category theory using the principle of naturality, and can be given an intuitive diagrammatic representation in terms of which all proofs are presented. In our main technical contribution, we use this diagrammatic representation to show that locally-applicable transformations on quantum channels are in one-to-one correspondence with deterministic quantum supermaps. This alternative characterization of quantum supermaps is proven to work for more general multiple-input supermaps such as the quantum switch and on arbitrary normal convex spaces of quantum channels such as those defined by satisfaction of signaling constraints. 
\end{abstract}

\section{Introduction}


Beyond the framework of standard quantum theory, in which states representing physical degrees of freedom incur changes over time, there is the framework of higher order quantum theory \cite{Chiribella2008TransformingSupermaps, Chiribella2008QuantumArchitecture, Chiribella2013QuantumStructure, Chiribella2009TheoreticalNetworks, Chiribella2010NormalOperations,  Bisio2019TheoreticalTheory, Kissinger2019AStructure, wilson2022Aprocesses, Perinotti2017CausalComputations}, in which dynamics themselves are transformed by higher order operations termed supermaps \cite{Chiribella2008TransformingSupermaps}. Quantum supermaps were originally defined to formalise the notion of a higher order map $S$ which may be applied to part of any two-input/two-output process $\phi$, as in the following intuitive picture: \[  \input{figs/informal_cp_3.tikz}.   \]  Such higher order operations were later generalised to act on constrained spaces, and were applied in the study of quantum information to analyze protocols in which quantum processes are treated as information-theoretic resources \cite{Kristjansson2020ResourceCommunication, GourDynamicalResources, Gour2021EntanglementChannel, Ebler2018EnhancedOrder, Procopio2019CommunicationScenariob, Procopio2020SendingOrders, Chiribella2021IndefiniteChannels, Chiribella2020QuantumOrders, Wilson2020ASwitches, Quintino2019ProbabilisticOperations, Quintino2019ReversingOperations, Dong2021Success-or-Draw:Computation, Guerin2016ExponentialCommunication, Chiribella2012PerfectStructures, Araujo2014ComputationalGates, Feix2015QuantumResource}. On top of providing a framework for formalising such protocols, supermaps are broad enough to incorporate higher order processes beyond those which can be interpreted as circuits with open holes \cite{Chiribella2008QuantumArchitecture, Roman2020OpenCalculus, Roman2020CombFeedback, Boisseau2022CorneringOptics, Hedges2017CoherenceGames, Riley2018CategoriesOptics, Ghani2016CompositionalTheory, Pollock2015Non-MarkovianCharacterisation, Hefford2022CoendCombs}, the canonical example of such a supermap being the quantum switch \cite{Chiribella2009QuantumStructure}. Acting on the space of non-signaling channels, the quantum switch uses a qubit to control the order in which wires are joined and as depicted in the following inuitive picture \[   \input{figs/infromal_switch1.tikz} \quad \bigoplus \quad \input{figs/infromal_switch2.tikz}, \] is often interpreted as a superposition of causal structures \cite{Chiribella2009QuantumStructure}. By being broad enough to incorporate examples such as the quantum switch, the supermap framework provides a way to study quantum causal structure as a resource in quantum information processing protocols \cite{Ebler2018EnhancedOrder, Procopio2019CommunicationScenariob, Procopio2020SendingOrders, Chiribella2021IndefiniteChannels, Chiribella2020QuantumOrders, Wilson2020ASwitches, Araujo2014ComputationalGates, Guerin2016ExponentialCommunication, Chiribella2012PerfectStructures, Araujo2017QuantumStructures, Renner2021ReassessingGates, Felce2020QuantumOrder}, and furthermore may find application in the study of quantum gravity \cite{Hardy1992QuantumTheories}. Indeed, the dependence of causal structure on mass distribution in general relativity and the possibility to superpose mass distributions in quantum theory, suggests the possibility of naturally occurring non-classical, or indefinite, causal structures \cite{Hardy2006TowardsStructure, Oreshkov2012QuantumOrder, Baumeler2013PerfectOrder, Baumeler2015TheOrder, Chiribella2009QuantumStructure, Castro-Ruiz2018DynamicsStructures, Baumann2021NoncausalCircuits, Costa2022AOrders, Silva2017ConnectingStates, Vilasini2022EmbeddingMatrices, Ormrod2022CausalSwitch, Felce2021IndefiniteTime} such as those present in the quantum switch. To formalise this intuition supermaps can be interpreted as modelling global spacetime structures as maps from interventions chosen in local quantum laboratories to probabilities \cite{Oreshkov2012QuantumOrder}. In this context, supermaps are typically expressed in the Choi representation, and are referred to as process matrices \cite{Oreshkov2012QuantumOrder}.  This abstract approach to modeling global spacetime structures allows to study causal structures beyond those which are even switch-like. These supermaps can sometimes  break causal inequalities \cite{Oreshkov2012QuantumOrder, Oreshkov2016CausalProcesses, Wechs2021QuantumOrder, Purves2021QuantumInequalityb, tein_inequalities, pinzani_topology_geometry}, which are an analogue of Bell inequalities for causal order.

The motivating picture of a supermap appears to only reference the possibility to draw processes in quantum theory as boxes with multiple inputs and outputs, and yet, current constructions and definitions of quantum supermaps \cite{Chiribella2008TransformingSupermaps, Kissinger2019AStructure, Oreshkov2012QuantumOrder} rely on additional mathematical structures. These structures are typically those of categories into which deterministic quantum theories embed, such as compact closure/Choi-Jamiolkowski isomorphism \cite{Choi1975CompletelyMatrices, Chiribella2008TransformingSupermaps, Kissinger2019AStructure} and convexity/coarse graining \cite{Oreshkov2012QuantumOrder}. This suggests that supermaps are in need of a more principled axiomatisation,  so that the conceptual grounds on which supermaps are understood match the formal grounds on which they are defined, and so that the entire framework of higher order quantum theory may be more easily lifted to arbitrary physical theories.

In this paper we provide such an axiomatisation, by building on the process-theoretic/categorical approach to quantum theory \cite{Coecke2017PicturingReasoning, Heunen2019CategoriesTheory, Abramsky2004AProtocols, Coecke2006KindergartenNotes, Coecke2010QuantumPicturalism, Selby2018ReconstructingPostulates} we discover that supermaps can indeed be axiomatised purely at the process-theoretic level, that is, with respect to symmetric monoidal structure \cite{Lane1971CategoriesMathematician}. From this result it follows that all of the operational physical principles such as compatibility with coarse-graining, convex combinations, linearity, and tensor extensions used in usual definitions of quantum supermaps can be viewed as consequences a simple principle: \[\textit{A supermap is the kind of function that can be applied locally}.\]The formalization of local applicability of supermaps that we use as our axiom can be understood in three consecutive steps. First, supermaps are functions on processes: \[ \input{figs/informal_cp_o.tikz} . \] Second, supermaps are equipped with extensions to functions on all two-input/two-output processes: \[ \input{figs/informal_cp_3.tikz} . \]Third, localization is enforced by requiring that such functions commute with actions on their extensions: \[ \input{figs/informal_cp_4aa.tikz} \quad = \quad \input{figs/informal_cp_4ab.tikz} .   \] 
In other words, supermaps are maps that can be applied locally to parts of processes. In diagrams such as those above, the wires connected to $S$ are those which $S$ is interpreted as having acted on, the remaining wires represent the potential environments that might be present when the supermap is implemented. We model these three principles for supermaps using a formal definition of \textit{locally-applicable transformation}. The definition can be stated in entirely diagrammatic terms, and this diagrammatic phrasing is used throughout as a toolbox for proving our main theorem. 

\begin{theorem*} Quantum supermaps are in one-to-one correspondence with locally-applicable transformations.
\end{theorem*}

Although we do not do so in the main text, many of the definitions and results can be neatly phrased in the language of category theory \cite{Lane1971CategoriesMathematician}. First, the presented correspondence between locally-applicable transformations and standard-definition quantum supermaps is  compositional, meaning that it can be phrased as an equivalence of categories \cite{Lane1971CategoriesMathematician}. Second, in the case of locally-well pointed theories such as quantum and classical information theory, the definition of locally-applicable transformation is simply that of a natural transformation between functors. As a consequence of this observation our main theorem can be summarised concisely as a categorical characterisation of supermaps on quantum theory: \[\textit{Quantum supermaps are equivalent to natural transformations.}\]Specialised instances of such natural transformations can in fact already be observed as playing a similar role to supermaps in the field of monoidal category theory, for instance in the definition of a traced monoidal category \cite{Joyal2021TracedCategories, Hasegawa97recursionfrom} or similarly in the formalisation of closed time-like curves \cite{Pinzani2019CategoricalTravel} where such curves are even termed super-operators. 

The results presented in this paper show that the $2^{nd}$-order of higher order quantum theory can be understood as a purely categorical/circuit-theoretic construction over $1^{st}$-order quantum theory. A first key open question is whether similar methods can be used to reconstruct all of higher order quantum theory \cite{Bisio2019TheoreticalTheory, Kissinger2019AStructure, Chiribella2008TransformingSupermaps, Oreshkov2012QuantumOrder} from purely compositional axioms. A second key open question, is whether this new axiomatization allows for easier generalisation of higher order quantum transformations to infinite dimensions and to non-monoidal physical frameworks. For instance supermaps have yet to be formalized within decompositional approaches to physics \cite{Gogioso2019ASpace, Arrighi2021QuantumTheory} or algebraic quantum field theories \cite{Haag1964AnTheory} where localizability of standard quantum transformations is of primary importance \cite{roberts_algebra, Doplichera2010TheChallenges, Brunetti2001ThePhysics, Strocchi2004RelativisticTheory} and infinite dimensionality is common. Ultimately, such further developments would bring the recently developed higher-order approach to the study of quantum causal structure closer to being applicable to theories of quantum gravity, where such structures have been predicted to play a key role \cite{Hardy2006TowardsStructure}.

\section{Preliminary Material}
To find a purely circuit-theoretic construction for supermaps we will need a notion of circuit-theory on which such a construction can be phrased. Here we present the model of circuit-theories as symmetric monoidal categories. We then review quantum supermaps, presenting them in terms of the graphical language for compact closed categories which captures the essence of the Choi-Jamiolkowski isomorphism \cite{Choi1975CompletelyMatrices}. Throughout this paper we use purple-shaded boxes to represent standard quantum processes, and white boxes to indicate parts of a diagram which should be interpreted as higher-order maps to be applied to standard processes.

\paragraph{Symmetric Monoidal Categories}
To ease the following presentation we present symmetric monoidal categories which are ``strict", meaning that equalities between objects are written in place of natural isomorphisms. For a formal treatment of non-strict symmetric monoidal categories the reader is referred to \cite{Lane1971CategoriesMathematician}, where it is noted that every symmetric monoidal category is suitably equivalent to a strict one. We from now on omit use of the word strict, leaving it as implicit for the remainder of the paper. Our working examples of symmetric monoidal categories will be the categories $\mathbf{QC},\mathbf{CP}$ of quantum channels and completely positive maps respectively. We use the term quantum channel to mean completely positive trace-preserving map. 

The first formal ingredient in the definition of a symmetric monoidal category $\mathbf{C}$ is the specification of objects of $\mathbf{C}$, which in quantum theory may be thought of as systems. In symmetric monoidal categories objects can be drawn as wires, for instance we draw an object $A$ as: \[  \input{figs/category_1.tikz}  . \]Objects of each of $\mathbf{QC}$ and $\mathbf{CP}$ are given by finite dimensional Hilbert spaces, which are used to represent quantum degrees of freedom. The second formal ingredient of a symmetric monoidal category is the assignment of a set  $\mathbf{C}(A,B)$ to each pair $A,B$ of objects. The set $\mathbf{C}(A,B)$ will be referred to as the set of of morphisms, and the elements $f \in \mathbf{C}(A,B)$ of this set are often denoted using the notation $f:A \rightarrow B$. Morphisms in symmetric monoidal categories can be drawn as boxes with input and output wires, used to represent their input and output objects: \[  \input{figs/category_2.tikz} . \]

In $\mathbf{QC}$ the morphisms $f:H_A \rightarrow H_B$ are completely positive trace preserving maps from $\mathcal{L}(H_A)$ to $\mathcal{L}(H_B)$. Similarly, in $\mathbf{CP}$ the morphisms are taken to be completely positive maps. Symmetric monoidal categories are next equipped with two compatible notions, sequential composition and parallel composition. Sequential composition of $f:A \rightarrow B$ and $g:B \rightarrow C$ can be denoted graphically by: \[  \input{figs/category_3.tikz}.  \] In quantum theory the sequential composition $g \circ f$ is typically interpreted as representing a process $g$ occurring some time after a process $f$. In $\mathbf{QC}$ and $\mathbf{CP}$ the composition rule is inherited directly from the standard notion of sequential composition for linear maps.

Parallel composition addresses both objects and morphisms, for each pair of objects $A,B$ a new object can be assigned called $A \otimes B$, typically interpreted as placing $A$ next to $B$. For each pair of morphisms $f:A \rightarrow A'$ and $g:B \rightarrow B'$ one can assign a new morphism $f \otimes g: A \otimes B \rightarrow A' \otimes B'$ typically interpreted as the concurrent action of $f$ and $g$. Diagrammatically $A \otimes B$ can be represented by placing wire $A$ next to wire $B$: \[  \input{figs/category_4.tikz}.  \]  The expression $f \otimes g$ can then be represented by the following diagram: \[ f \otimes g \quad := \quad \tikzfigscale{1}{figs/process_2} . \] In our working examples the parallel composition of objects $H_A,H_B$ is given by the standard tensor product $H_A \otimes H_B$ of finite dimensional Hilbert spaces. A key feature of symmetric monoidal categories which allows them a diagrammatic calculus is the interchange law \[  (f \circ f') \otimes (g \circ g') = (f \otimes g) \circ (f' \otimes g') , \] which implies unambiguous interpretation of the following diagram: \[  \input{figs/category_int.tikz}  . \] In a monoidal category the parallel composition is required to be associative so that $(A \otimes B) \otimes C = A \otimes (B \otimes C)$, this is indeed true for the tensor product of Hilbert spaces up to natural isomorphism $H_A \otimes (H_B \otimes H_C) \cong (H_A \otimes H_B ) \otimes H_C$ as addressed in the introductory remarks. A monoidal category furthermore comes equipped with a notion of empty space given by an object $I$ satisfying $A \otimes I = A = I \otimes A$, this object in our motivating examples is given by the Hilbert space $\mathbb{C}$ satisfying $H_A \otimes \mathbb{C} \cong H_A \cong \mathbb{C} \otimes H_A$ for every Hilbert space $H_A$. Making use of the notion of empty space $I$, the morphisms of type $\rho:I \rightarrow A$ are often interpreted states $A$. This is well motivated when $I$-wires are omitted from diagrams so that morphisms of type $I \rightarrow A$ are drawn as having only output wires: \[  \input{figs/category_state.tikz} . \] Finally, a symmetric monoidal category comes equipped with a swap-morphism $\beta: A \otimes B \Rightarrow B \otimes A$ depicted by: \[  \input{figs/category_6.tikz} , \] satisfying a variety of natural properties \cite{Lane1971CategoriesMathematician}. In the course of this paper, we will need to address two refined structures which supermaps are typically defined with respect to, those monoidal structures which are additionally causal and those which are additionally compact closed.

\paragraph{Causality: The Trace Channel}
A symmetric monoidal category is causal \cite{Coecke2017PicturingReasoning, Coecke2013CausalProcesses} if its unit object is terminal, meaning in concrete terms that for each $A$ the set $\mathbf{C}(A,I)$ is a singleton. In a causal symmetric monoidal category we typically represent the unique effect of type $A \rightarrow I$ using the discard symbol: \[  \input{figs/discard.tikz}_{A} . \] Our working example $\mathbf{QC}$ is a causal category, with the unique morphism of type $H_A \rightarrow \mathbb{C}$ given by the trace quantum channel. When a category is causal a direction of time is essentially fixed as propagating up the page. Indeed, with no way to vary between effects of type $\mathbf{C}(A,I)$ then there is no way to send information in the opposite direction, from the top of the page to the bottom.
\paragraph{Compact Closure: Channel-State Duality}
A symmetric monoidal category is compact closed if it is equipped for each $A$ with an object $A^{*}$ and morphisms $\cup_{A}:I \rightarrow A^{*} \otimes A$ and $\cap_{A}:A \otimes A^{*} \rightarrow I$ which amongst other natural conditions \cite{Coecke2017PicturingReasoning} satisfy $(\cap \otimes id) \circ (id \otimes \cup) = id$, graphically this reads as: \[  \input{figs/compact_1.tikz}  \quad = \quad \input{figs/compact_2.tikz}.  \]When compact closure is present it offers significant flexibility by giving an internalised way to connect input and output ports of circuit diagrams together. Our working example $\mathbf{CP}$ is a compact closed category with ${H_A}^{*}$ given by the dual Hilbert space to $H_A$ and the cup $\cup_{H_{A}}:\mathbb{C} \rightarrow {H_A}^{*} \otimes H_{A}$ given by the maximally entangled Bell-state. Similarly the cap is given by the maximally entangled Bell-effect. The Choi-isomorphism \cite{Choi1975CompletelyMatrices}, which provides an identification between completely positive maps and positive operators, is given by applying completely positive maps to Bell-states and so can be expressed graphically as: \[ \input{figs/category_2.tikz} \quad \longleftrightarrow \quad \input{figs/category_7.tikz} . \]

\paragraph{Additive Structure} The symmetric monoidal category $\mathbf{CP}$ is equipped with a notion of closure under positive linear combination, for any pair $\phi_0,\phi_1 \in \mathbf{CP}(A,A')$ of completely positive maps and pair $a_0,a_1 \in \mathbb{R}_{+}$ of positive real numbers a new combined completely positive linear map $a_0 \phi_0 + a_1 \phi_1 \in \mathbf{CP}(A,A')$ can be constructed. The symmetric monoidal category $\mathbf{QC}$ inherits a more restricted form of closure from $\mathbf{CP}$ termed closure under convex combinations, meaning that for any pair $\phi_0,\phi_1 \in \mathbf{QC}(A,A')$ of quantum channels and any probability $p \in [0,1]$ then the linear combination $p\phi_0 + (1-p)\phi_1 \in \mathbf{QC}(A,A')$ is itself a quantum channel.

\paragraph{Constrained Spaces}
In this paper, beyond supermaps on entire channels, we will aim to reconstruct supermaps which can be applied to channels satisfying signalling constraints  \cite{Eggeling2002SemicausalSemilocalizable, perinotti_causal_influence}. These quantum supermaps are the part of the higher-order toolbox typically used to study quantum causal structure. Let us begin with a simple example of a signalling constraint given by the specification of the set of one-way signaling channels, we denote the statement that $\phi$ be a one-way signaling channel by:  \[ \input{figs/phi.tikz} \quad \in \quad \mathcal{E}_{sig}\Big(\input{figs/phi_onesig.tikz}\Big)  . \] The depicted graph indicates that $\phi$ may not transmit information from the bottom right wire to the top left wire, this constraint is expressed concretely by the existence of a channel $\phi'$ satisfying the following equation: \[  \input{figs/phi_onesig1.tikz} \quad =  \quad  \input{figs/phi_onesig2.tikz}  .  \] In general more complex signalling constraints can be encoded by relations \cite{Lorenz2020CausalTransformationsb, Wilson2021ComposableConstraints}. For instance the statement \[ \input{figs/phi_three.tikz}  \in \quad   \mathcal{E}_{sig} \Big( \input{figs/phi_three_cone.tikz} \Big)  ,  \] can be used to encode the statements  \[ \input{figs/phi_three_discard.tikz} = \input{figs/phi_three_consiga.tikz}, \quad  \textrm{and} \quad \input{figs/phi_three_discard2.tikz} = \input{figs/phi_threee_consig2a.tikz}.  \] A formal definition for arbitrary relations is given in \cite{Wilson2021ComposableConstraints}. Useful for our reconstruction will be the observation that for any relation $\tau$ the set $\mathcal{E}_{sig}(\tau)$ of processes satisfying constraint $\tau$ has the convenient property of being closed under convex combinations. Another convenient property of signaling constraints in quantum theory for our reconstruction is their equivalence to localizability constraints, which we refer to here as pathing constraints so as to not confuse with our main definition of local applicability. As an example of a pathing constraint, the one-way signalling condition in quantum theory is equivalent to requiring the following decomposition: \[ \input{figs/phi.tikz} \quad = \quad  \input{figs/path_onesig.tikz} . \] This decomposition abstracts the impossibility of signalling from the bottom right wire to the top left wire by forbidding the existence of a vertically directed path between them. Since in quantum theory pathing constraints are equivalent to signaling constraints \cite{Eggeling2002SemicausalSemilocalizable}, pathing constraints in quantum theory are also closed under convex combinations. Pathing constraints come with the advantage however that they can be phrased in arbitrary symmetric monoidal categories without reference to a prioritised effect or time-direction. In intuitive terms pathing constraints abstract from causality to compositionality, as an example of a more complex pathing constraint consider the statement: \[ \input{figs/phi_three.tikz}  \in \quad  \mathcal{E}_{path} \Big( \input{figs/phi_three_cone.tikz} \Big),  \] this is a concise encoding of the statement that $\phi$ can be decomposed in such a way as to not provide a vertically directed path from the left-most input to the right-most output, or vice-versa, i.e \[ \input{figs/phi_three.tikz} = \input{figs/path_1.tikz} \quad  \textrm{and} \quad \input{figs/phi_three.tikz} = \input{figs/path_2.tikz}.  \] The equality $\mathcal{E}_{path}(\tau) = \mathcal{E}_{sig}(\tau)$ will allow us to reconstruct quantum supermaps on signaling constraints as locally-applicable transformations applied to their corresponding pathing constraints. Phrased in this way the quantum supermaps used to study quantum causal structure, are recovered from entirely compositional definitions with no reference to causality or chosen direction of the flow of time.  

\paragraph{Extended Process Sets}
In this paper we will be often concerned with extending sets of processes to auxiliary systems. For a channel to be treated as an extension of channels from some set, one ought to expect that no internal local dynamics on the auxiliary extended systems should be able to change that fact. Consequently we introduce the following minimum requirement for extensions of a set $K$.
\begin{definition}
A family of sets $K_{X,X'} \subseteq \mathbf{C}(A \otimes X,A' \otimes X')$ is an extension set for $K \subseteq \mathbf{C}(A,A')$ if $K_{I,I'} = K$ and for every $\phi \in K_{X,X'}$, $f:Y \rightarrow X \otimes E$, and $g:X' \otimes E \rightarrow Y'$ then \[\input{figs/ext_0.tikz} \quad \in \quad K_{Y,Y'} .  \]
\end{definition}
The examples we will make use of in the following text are those given by \textit{dilation extension}. The principle is to define, for any set $K \subseteq \mathbf{C}(A,A')$ and pair $X,X'$, the extension by $X,X'$ of $K$ to be the set of all processes $\phi \in \mathbf{C}(A \otimes X, A' \otimes X')$ which are stinespring dilations of processes of $K$.
\begin{definition}
For each $K \subseteq  \mathbf{C}(A,A')$ and pair $(X,X')$ the dilation extension by $X,X'$ denoted $\textnormal{\texttt{dExt}}_{X,X'}(K)$ is the subset of $\mathbf{C}(A \otimes X,A' \otimes X')$ given by: \[\phi \in \textnormal{\texttt{dExt}}_{X,X'}[K]  \iff \forall \rho, \sigma: \quad  \input{figs/ext_1.tikz} \quad \in \quad K  . \]
\end{definition}
Note that for the example $K = \mathbf{C}(A,A')$ then the extension $\textnormal{\texttt{dExt}}_{X,X'}(K)$ returns the entire set of two-input/two-output morphisms $\mathbf{C}(A \otimes X,A' \otimes X') $\footnote{It is natural to wonder whether this holds in the case in which a symmetric monoidal category has no states or effects, in-fact it does so vacuously (since the requirement is that the condition be satisfied for all such states and effects). }. 
Whilst this holds for any symmetric monoidal category, dilation extensions for other non-trivial subsets are conceptually better suited to categories that are expected to have non-trivial states and effects, such as operational probabilistic theories \cite{Chiribella2016QuantumPrinciples}, rather than categories of pure reversible evolutions, such as the category of unitaries.

For a causal symmetric monoidal category the extended set can be rephrased in the following way: \[\phi \in \textnormal{\texttt{dExt}}_{X,X'}[K]  \iff \forall \rho: \quad  \input{figs/ext_2.tikz} \quad \in \quad K  , \] which is the form used in \cite{Chiribella2009QuantumStructure} to define extensions to subsets of quantum channels. In the appendix we note that such extended channel sets can be viewed in categorical language as functors into the category of sets. Whenever a subset $K \subseteq \mathbf{QC}(A,A')$ is closed under convex combinations then it follows that $\textnormal{\texttt{dExt}}_{X,X'}(K)$ is closed under convex combinations, we will from now on rephrase the statement that a set $K$ be closed under convex combinations as simply the statement that $K$ be convex.


\paragraph{Quantum Supermaps: Standard Definition}
In this paper we use the category $\mathbf{CP}$ of completely positive maps to express the definition of supermap, this is sufficient for our purposes since the sets we choose to work with are internal \cite{Chiribella2009QuantumStructure, Chiribella2013QuantumStructure}. We follow the presentation of \cite{Kissinger2019AStructure} in which a completely positive linear map \[ \tikzfigscale{0.8}{figs/newsupermap_1}  ,\] is informally interpreted as a diagram with a hole, a function which accepts a channels of type $A \rightarrow A'$ and returns a channel of type $B \rightarrow B'$ by using compact closure or Choi-Jamiolkowski isomorphism \cite{Choi1975CompletelyMatrices} \[ \tikzfigscale{0.8}{figs/newsupermap_2}.  \] 
This notation puts us in a position to concisely phrase the standard definition method for quantum supermaps.
\begin{definition}
Let $\mathbf{C} \subseteq \mathbf{P}$ be an inclusion of a symmetric monoidal category $\mathbf{C}$ into a compact closed category $\mathbf{P}$ and let $K_{-,=}$ and $M_{-,=}$ be extension sets for $K \subseteq \mathbf{C}(A,A')$ and $M \subseteq \mathbf{C}(B,B')$ respectively. A $\mathbf{P}$-supermap on $\mathbf{C}$ of type $S:K_{-,=}  \rightarrow M_{-,=}$ is a morphism in $\mathbf{P}$ of type $S: A^{*} \otimes A' \rightarrow B^{*} \otimes B'$ such that for every $\phi \in K_{X,X'}$ then  \[ \tikzfigscale{0.8}{figs/newsupermap_ext} \quad \in \quad M_{X,X'} . \]
\end{definition}
In short $\mathbf{P}$-supermaps are processes from $\mathbf{P}$ which can act on processes from $\mathbf{C}$. Furthermore, in the spirit of the axiomatic definition of completely positive trace preserving maps, it is required that the $\mathbf{P}$-supermap sends processes from $\mathbf{C}$ to processes from $\mathbf{C}$ even when acting only on part of the process.


$\mathbf{P}$-supermaps can be composed in sequence, with this sequential composition inherited directly from $\mathbf{P}$, in categorical terms this means that the $\mathbf{P}$-supermaps define a category $\mathbf{Psup}[\mathbf{C}]$. For brevity we will use the term ``quantum supermap" for $\mathbf{CP}$-supermap on $\mathbf{QC}$ and denote $\mathbf{CPsup}[\mathbf{QC}]$ by $\mathbf{QS}$. Note that we are working with quantum supermaps applied to arbitrary subsets here, including those which satisfy signalling constraints. For the study of indefinite causal structure the most prevalent example is that of supermaps on non-signalling channels and their dilation extensions: \[ \input{figs/phi.tikz} \quad \in \quad \mathcal{E}_{sig}\Big(\input{figs/phi_nosig.tikz}\Big),  \quad \quad \quad \quad \input{figs/phiext.tikz} \quad \in \quad \textnormal{\texttt{dExt}}_{X,X'}\Big[ \mathcal{E}_{sig}\Big(\input{figs/phi_nosig.tikz}\Big) \Big]  .  \]

The lack of communication between wires in such channels allows them to be combined in a variety of ways without producing time-loops, so that the output may still be guaranteed to be a deterministic channel. Quantum supermaps of type\footnote{We use the symbols \texttt{-} and \texttt{=} to denote free-variables.} \[  \textnormal{\texttt{dExt}}_{\texttt{-},\texttt{=}}\Big[ \mathcal{E}_{sig}\Big(\input{figs/phi_nosig.tikz}\Big) \Big] \rightarrow \mathbf{C}(A \otimes \texttt{-}, Q \otimes A \otimes \texttt{=})  \] which have been the subject of considerable study are switches, which take channels and by reference to a control state, plug them together in a combination of orders:  \[ \tikzfigscale{0.8}{figs/switch_1a} \quad = \quad \tikzfigscale{0.8}{figs/switch_2a} \quad + \quad \tikzfigscale{0.8}{figs/switch_3a} . \] Note that when the above diagrams are drawn in the category $\mathbf{CP}$ they represent classical mixtures of causal orders. Instead, if one interprets the above diagram in $\mathbf{fHilb}$ and \textit{then} embeds into $\mathbf{CP}$ via the usual doubling functor with $\mathcal{F}(f)(\rho) := f (\rho) f^{\dagger}$ then interference is included and what is recovered is the \textit{quantum} switch. 

In the literature on higher-order quantum processes it is common to think of switches (and more generally the supermaps on the non-signalling channels) as being multiparty. More precisely, rather than thinking of supermaps as having one constrained input, they are thought of and pictorially depicted as having many unconstrained inputs. The informal intuitive picture that should be kept in mind for such maps is the following \[ \input{figs/informal_cp_3_multi.tikz} . \]

Again, when a theory comes equipped with an inclusion $\mathbf{C} \subseteq \mathbf{P}$ into a compact closed category one can find a quite direct axiomatisation of such supermaps. We refer to a supermap as multiparty when we think of its input type as a list, as in the following definition.
\begin{definition}
Let $\mathbf{C} \subseteq \mathbf{P}$ be an inclusion of a symmetric monoidal category $\mathbf{C}$ into a compact closed category $\mathbf{P}$, and let $K^{i}_{X,X'} \subseteq \mathbf{C}(A_i \otimes X, A_i' \otimes X')$, a morphism \[ \tikzfigscale{0.6}{figs/newmultisupermap1} , \] in $\mathbf{P}$ is a $\mathbf{P}$-supermap on $\mathbf{C}$ of type $S:K^1_{-,=} \dots K^n_{-,=} \rightarrow M_{-,=}$ if and only if for every family ${\phi}_{i} \in K^i_{ X_i ,  X_i'}$ then  \[ \tikzfigscale{0.6}{figs/newmultisupermap2} \quad \in \quad  M_{X_1 \otimes \dots \otimes X_n , X_1' \dots \otimes X_n'}. \]
\end{definition}
Note that we can (as we have above) allow each of the independent intputs in the list to actually be constrained sets of processes. From now on, we will keep in mind the case in which each input is unconstrained so that supermaps are required to be processes in $\mathbf{P}$ which are well-behaved on the space of all product channels. The supermaps on the non-signalling channels are the same as supermaps on the product channels, this can be seen for instance in the definition of the non-signalling channels as the double closure of the set of products of channels \cite{Kissinger2019AStructure}. As a consequence of this, in the case of either quantum theory or classical theory, the multiparty supermaps are equivalent to the supermaps on the space of non-signalling channels.

\paragraph{Summary}
In quantum information theory, supermaps on quantum channels are defined by using channel state duality (compact closure) of the theory of completely positive maps into which they embed. In this paper we ask the following: \textit{What can we say when background compact closure cannot be assumed?} In other words:  \textit{Can supermaps be characterized in terms of sequential and parallel composition alone?}. We answer positively, showing that quantum supermaps could have been defined all along by an abstract principle of local applicability (in category-theoretic language termed naturality). After motivating and defining locally-applicable transformations on arbitrary symmetric monoidal categories we show a one-to-one correspondence between them and quantum supermaps when applied to the symmetric monoidal category of quantum channels.

\section{Formalisation of Local Applicability}
Our goal is to find an axiom for quantum supermaps which can be applied to \textit{any} symmetric monoidal category. In this section we show how to so, by defining higher order functions and requiring the existence of extensions for all auxiliary systems which further commute with actions on those auxiliary systems. We begin by warming up to the definition by observing a notion of local applicability present in monoidal categories, such as standard quantum theory, which we plan to abstract to higher order functions.  
\paragraph{Local Applicability in Standard Quantum Theory}
The category $\mathbf{QC}$ of quantum channels is a symmetric monoidal category, one consequence of the parallel composition rule $f \otimes g$ is that it gives a way to view any $f:A \rightarrow A'$ as locally applicable in an intuitive sense. We discuss this locality principle in the monoidal setting and then comment on how it can be abstracted to a statement about locality of functions built from such an $f$ on state sets. To begin, consider a channel $f:A \rightarrow A'$ from $A$ to $A'$:
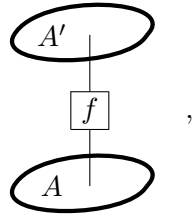
\[   \input{figs/apply1.tikz} , \]
whenever $A$ can be viewed as part of a larger system $N = A \otimes X$ then $f$ can be locally applied to $N$ by using $f \otimes id_X: N \rightarrow A' \otimes X$ \[   \input{figs/apply2.tikz}.  \] Crucial to the interpretation of locality in $f \otimes id_X$ is that $f \otimes id_X$ commutes with all actions on $X$, this follows in this case by the interchange law for monoidal categories $(f \otimes id_{X'} ) \circ (id_{A} \otimes g) = (id_{A'} \otimes g) \circ (f \otimes id_{X})$: \[   \input{figs/apply3.tikz} \quad \quad = \quad \quad \input{figs/apply4.tikz}. \] In other words, a general monoidal category gives a collection of morphisms, all of which can be viewed as being locally applicable, in an informal sense. 

A consequence of this local applicability is the possibility to construct from $f$ a family of functions on states $l(f)_X:\mathbf{C}(I,A \otimes X) \rightarrow \mathbf{C}(I,A' \otimes X)$ which exhibit the local applicability of $f$. Explicitly, by using tensor extensions with the identity the function $l(f)_{X}(\rho_{AX}) := (f \otimes id_X)(\rho_{AX})$ can be defined for each $X$. The abstract functions $l(f)_{X}$ which represent the action of $f$ on states indeed inherit a notion of local applicability from $f$. The functions $l(f)_X$ can be seen to leave the environment system $X$ untouched in the sense that the action of any $g$ on $X$ commutes with the application of the function $l(f)_X$. The above sentence is captured in formal terms by the equation $l(f)_{X'}(id_A \otimes g(\rho_{AX})) = (id_{A'} \otimes g) (l_{X}(\rho_{AX}))$ which is guaranteed to hold for any $g:X \rightarrow X'$ since \begin{align*}
    l(f)_{X'}(id_A \otimes g(\rho_{AX})) 
    = & (f \otimes id_{X'}) \circ (id_A \otimes g)(\rho_{AX}) \\
    = &  (id_{A'} \otimes g) \circ (f \otimes id_{X})(\rho_{AX})\\
    = & (id_{A'} \otimes g) (l_{X}(\rho_{AX})) 
\end{align*} 
We now give an axiom which re-characterises quantum supermaps by generalizing this concept of local applicability of functions on states to local applicability of functions on processes. The only instances of locally-applicable transformations on quantum channels will turn out to be those which are simulated by the standard-definition quantum supermaps of \cite{Chiribella2008TransformingSupermaps}. We split the motivations for the definition of locally-applicable transformation into three consecutive principles.

\paragraph{Principle 1: Supermaps are Functions on Processes}
The kind of picture usually drawn with the aim of capturing diagrammatically the concept of a supermap from the space of processes $\mathbf{C}(A, A')$ to the space of processes $\mathbf{C}(B , B')$ is some variation of the following \[ \input{figs/informal_cp_o.tikz} . \] As such our first step to characterising supermaps of type $\mathbf{C}(A,A') \rightarrow \mathbf{C}(B,B')$ is to consider functions of the same type $\mathbf{C}(A,A') \rightarrow \mathbf{C}(B,B')$. More generally for $K \subseteq \mathbf{C}(A,A')$ and $M \subseteq \mathbf{C}(B,B')$ the first step is to consider functions $S:K \rightarrow M$ from the set $K$ to the set $M$.

\paragraph{Principle 2: Supermaps Can be Extended to Functions on all two-input/two-output Processes} When we say that we wish for the map $S:\mathbf{C}(A,A') \rightarrow \mathbf{C}(B,B')$ to be locally applicable, we mean that we wish to formalise the following picture: \[\input{figs/informal_cp_3a.tikz}. \]The next step toward such a formalisation is to specify for each $X,X'$ the action of $S$ when applied to the $A,A'$ part of any morphism $\phi \in \mathbf{C}(A \otimes X, A' \otimes X')$. Consequently we say that a locally-applicable transformation must be equipped with a family of extended functions $S_{XX'}: \mathbf{C}(A \otimes X, A' \otimes X' ) \longrightarrow \mathbf{C}(B \otimes X, B' \otimes X' )$ for every $X,X'$. For the generalised case of supermaps of type $S: K \rightarrow M$ we instead require the specification of a function $S_{X,X'}:K_{X,X'} \rightarrow M_{X,X'}$ for every $X,X'$ with $K_{X,X'}$ some extension set for $K$ and similarly for $M$. 
For readability we will from now on notate the action of such a family of functions in the following way \[ S_{X,X'}(\phi) \quad :=     \quad \input{figs/informal_cp_3b.tikz},   \] where the dotted lines express the idea that the wires they connect are to be interpreted as auxiliary systems. Formally, the dotted lines simply allow us to denote diagrammatically which systems are the $X$ and $X'$ of $S_{X,X'}$. Conceptually, the dotted wires for $X,X'$ are intended to be indicate that $S_{X,X'}$ ought not act on them. At this moment of the formalization however, we have not imposed any additional mathematical condition on these diagrams which encodes this concept. 

\paragraph{Principle 3: Supermaps Commute With Actions on Their Extensions} A key feature of a local operation is commutation with operations applied to auxiliary spaces, we generalise this notion of locality to input-output operations, informally we aim to capture the equivalence of the following two pictures: \[ \input{figs/informal_cp_4aa.tikz} \quad = \quad \input{figs/informal_cp_4ab.tikz}, \] which can be formalised for functions on simple types $\mathbf{C}(A,A') \rightarrow \mathbf{C}(B,B')$ or on generalised types $K \rightarrow M$.%
\begin{definition}[locally-applicable transformations] Let $K_{-,=}$ and $M_{-,=}$ be extension sets,
a locally-applicable transformation of type $S:K_{-,=} \longrightarrow M_{-,=}$ on a symmetric monoidal category $\mathbf{C}$ is a family of functions $S_{XX'}:K_{X,X'} \rightarrow M_{X,X'}$ such that for every $g:X'\otimes Z \rightarrow Y'$, $f:Y \rightarrow X \otimes Z $, and $\phi:A \otimes X \rightarrow A' \otimes X'$ then 
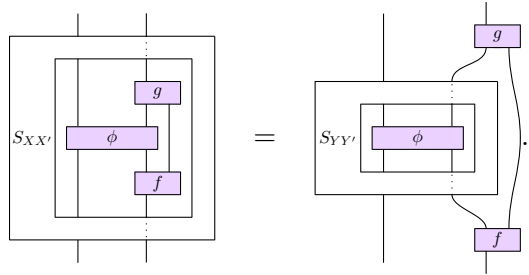
\[ \input{figs/formal_cp_1aa.tikz} \quad = \quad \input{figs/fornal_cp_1bb.tikz}.  \]
\end{definition}
The above definition is equivalent to the requirement of the following distinct rules of \textit{naturality} \[ \input{figs/formal_cp_1a.tikz} \quad = \quad \input{figs/formal_cp_1b.tikz},  \] and \textit{dragging}  \[  \input{figs/drag1.tikz} \quad =  \quad \input{figs/drag2.tikz} . \] where in the specific case of interest of the category $\mathbf{QC}$ of quantum channels, only the first condition need actually be given, since
for all causal $\rho$ we have the following equation by naturality: \[  \input{figs/tensor1.tikz} \quad =  \quad \input{figs/tensor2.tikz}, \] and following further equations, again by naturality:  \[ =  \quad \input{figs/tensor3.tikz} \quad =  \quad \input{figs/tensor4.tikz}  .  \] Together these entail box-dragging since quantum theory has \textit{enough causal states} \cite{Kissinger2019AStructure}. In the appendix we note that consequently locally-applicable transformations on $\mathbf{QC}$ can be phrased in the language of category theory as natural transformations, motivating our use of the word ``naturality". We further note that locally-applicable transformations can be composed, given locally-applicable transformations $S:K_{-,=} \rightarrow M_{-,=}$ and $T:M_{-,=} \rightarrow N_{-,=}$ one can construct the locally-applicable transformation $(T \circ S):K_{-,=} \rightarrow N_{-,=}$ by defining for each $X,X'$ the functions $(S \circ T)_{XX'}(\phi) = S_{XX'}(T_{XX'}(\phi))$. This compositionality of locally-applicable transformations can be phrased in the language of category theory by stating that they form a category, which we denote by $\mathbf{lot}[\mathbf{C}]$.

\section{Examples} We now consider a series of constructive examples of supermaps beginning with those which can be guaranteed to exist on \textit{any} symmetric monoidal category $\mathbf{C}$. We begin with combs in arbitrary symmetric monoidal categories as defined in \cite{Coecke2014AResources}.
\begin{example}[Combs]
For every symmetric monoidal category $\mathbf{C}$ and pair of morphisms $a:A \rightarrow E \otimes B$ and $b: E \otimes A' \rightarrow B'$ one can define a locally-applicable transformation of type $\mathbf{C}(A-,A'=) \rightarrow \mathbf{C}(B-,B'=)$ by   \[  \input{figs/final10.tikz} \quad := \quad \input{figs/rep1.tikz}  . \] Indeed, note that: \[  \input{figs/formal_cp_1aa.tikz} \quad = \quad \input{figs/rep2.tikz} \quad = \quad \input{figs/rep3.tikz} \quad = \quad \input{figs/fornal_cp_1bb.tikz} . \]   
\end{example}
We from now on refer to such a locally-applicable transformation by $\texttt{comb}[a,b]$ and its components by $\texttt{comb}[a,b]_{X,X'}$, such combs form a subcategory of $\mathbf{lot}[\mathbf{C}]$ which is isomorphic as a category to $\mathbf{comb}[\mathbf{C}]$ as defined in \cite{Hefford2022CoendCombs}. This example can be generalised to combs of type $\texttt{comb}[a,b]:K_{-,=} \rightarrow M_{-,=}$ which are those which furthermore satisfy $\phi \in K_{X,X'} \implies \texttt{comb}(a,b)(\phi) \in M_{X,X'}$. These examples can be further generalised to combs from categories into-which $\mathbf{C}$ embeds, as opposed to combs from $\mathbf{C}$ itself.
\begin{example}
For every pair $\mathbf{C},\mathbf{D}$ of symmetric monoidal categories with $\mathbf{C} \subseteq \mathbf{D}$ one can define the $\mathbf{D}$-combs on $\mathbf{C}$ of type $K \rightarrow M$ to be the combs of $\mathbf{D}$ which preserve morphisms of $\mathbf{C}$. Formally, that is, the transformations of type $\texttt{comb}[a,b]$ on $\mathbf{D}$ such that for all $\phi \in K_{X,X'}$ then $\texttt{comb}[a,b]_{X,X'}(\phi) \in M_{X,X'}$. 
\end{example}
For a compact closed category $\mathbf{P}$ with $\mathbf{C} \subseteq \mathbf{P}$ the notions of $\mathbf{P}$-supermap and $\mathbf{P}$-comb on $\mathbf{C}$ are equivalent. Through this equivalence, $\mathbf{P}$-supermaps always give examples of locally-applicable transformations.
\begin{lemma}
Let $\mathbf{P}$ be a compact closed category and $\mathbf{C}$ be a symmetric monoidal category, there is a one-to-one correspondence between the $\mathbf{P}$-combs on $\mathbf{C}$ and the $\mathbf{P}$-supermaps on $\mathbf{C}$.
\end{lemma}
\begin{proof}
Let $S$ be a $\mathbf{P}$-comb on $\mathbf{C}$ of type $K_{-,=} \rightarrow M_{-,=}$ then one can construct \[  \tikzfigscale{0.6}{figs/newsupermap_1} \quad := \quad \tikzfigscale{0.6}{figs/final_comb} \quad = \quad \tikzfigscale{0.8}{figs/rep5}, \] which indeed is a $\mathbf{P}$-supermap since \[ \tikzfigscale{0.6}{figs/final9} \textrm{ } = \quad \input{figs/rep6.tikz} \quad = \quad \input{figs/rep1.tikz} . \] Instead let $S$ be a $\mathbf{P}$-supermap then one can construct the locally-applicable transformation \[  
\mathcal{F}(S)_{X,X'} \quad :=  \quad  \tikzfigscale{0.6}{figs/rep7} . \] These two constructions are furthermore inverse to each-other. The assignment is furthermore a functor $\mathcal{F}_{\mathbf{P},\mathbf{C}}:\mathbf{Psup}[\mathbf{C}] \rightarrow \mathbf{lot}[\mathbf{C}]$ meaning in concrete terms that $\mathcal{F}$ preserves composition and identities.
\end{proof}
The above story and equivalence between supermaps made with combs or morphisms with compact closure can be generalised to embeddings which are weak in the sense of $2$-category theory, a discussion of this point is given in the appendix, where it is noted that this generalisation allows to use the compact-closed category $*\mathbf{Hilb}$ \cite{Gogioso2017Infinite-dimensionalMechanics} to define a variety of locally-applicable transformations on $\mathbf{sepU}$, the category of unitaries between seperable Hilbert spaces.

Locally-applicable transformations in short state a bare-minimum requirement expected of quantum supermaps, satisfied by a variety of more familiar examples. A clear difficulty in the definition of supermaps, is the variety of potential definition methods. By using the unifying principle of a locally-applicable transformation, as a minimum requirement for supermaps, we will in fact find that all possible definitions of supermaps on finite dimensional quantum theory are equivalent by characterizing all locally-applicable transformations on $\mathbf{QC}$ as quantum supermaps, by which we mean $\mathbf{CP}$-supermaps on $\mathbf{QC}$.

\section{Characterisation of Standard Quantum Supermaps}
We now prove that locally-applicable transformations on finite dimensional quantum channels are equivalent to standard-definition quantum supermaps. We have already proven that all quantum supermaps define locally-applicable transformations, so what remains is to prove that this assignment can be inverted, that every locally-applicable transformation defines quantum supermap. We begin by identifying a key feature of quantum theory, the existence of control for convex sets. We then give our proof in three steps, proving inheritance of convex linearity for locally-applicable transformations on convex sets, proving their unique extension to action on completely positive maps, and finally proving their realisation in terms of standard quantum supermaps. We begin by addressing the property of control.
\begin{definition}[Control]
A set $K$ has control if for every pair $\phi_{0},\phi_1 \in \textnormal{\texttt{dExt}}_{X,X'}(K)$ there exists $\phi \in  \textnormal{\texttt{dExt}}_{X \otimes Y,X' \otimes Y'}(S)$ and a pair of states $\rho_0, \rho_1$ such that: \[    \input{figs/control_1_caus.tikz} . \]
\end{definition}
Conveniently convex sets in $\mathbf{QC}$ are always controlled, in fact the existence of control is equivalent to asking for closure under convex combinations.
\begin{lemma}
A set $K \subseteq \mathbf{QC}(A,A')$ has control if and only if it is convex.
\end{lemma}
\begin{proof}
We begin by showing that whenever $K$ is convex it has control. Note that whenever $K$ is convex then $\textnormal{\texttt{dExt}}_{X,X'}(K)$ is convex for every choice of $X$ and $X'$. Now choose a pair of states $\rho_0,\rho_1 \in \mathbf{QC}(I,Y)$ on an object $Y$ which are distinguishable in the sense that there exist effects $e_0,e_1 \in \mathbf{CP}(Y,I)$ satisfying $e_i \circ \rho_j = \delta_{ij}$ and then construct the following process $\phi$ \[ \input{figs/control_proof0.tikz} \quad := \quad  \input{figs/control_proof1.tikz} \quad + \quad  \input{figs/control-proof1b.tikz} ,  \] certainly by inserting $\rho_0,\rho_1$ into the rightmost wire the channels $\phi_0,\phi_1$ are recovered, what remains is to show that $\phi $ is in $\textnormal{\texttt{dExt}}_{X,X'}(K)$. Indeed consider checking the reduction of $\phi$ given by applying an arbitrary state and effect of $\mathbf{QC}$ to its auxiliary wires, given that $\mathbf{QC}$ is causal this is given by: \[ \input{figs/control_proof_2a.tikz} \quad + \quad  \input{figs/control_proof2b.tikz}  \] for some $\tau$. Now note that each of the post selected states is a normalised state equipped with a probability, for instance: \[ \input{figs/norm_0.tikz} \quad = \quad \input{figs/norm_1.tikz} \quad \times  \quad \frac{\mbox{\huge  1}}{ \input{figs/norm_1.tikz}} \quad \times \quad  \input{figs/norm_0.tikz}  \quad = \quad p(e_0|\tau) \times \tau_{|0} \] and similarly for $e_1$. Since $\phi_0,\phi_1$ are elements of $\textnormal{\texttt{dExt}}_{X,X'}(K)$ it then follows from the above that the reduction of $\phi$ is a convex combination of elements of $K$, explicitly the application of arbitrary state and effect to $\phi$ returns: \[ p(e_0|\tau) \quad \input{figs/norm_control1.tikz} \quad + \quad  p(e_1|\tau) \quad \input{figs/norm_control2.tikz} \] where $p(e_0|\tau) + p(e_1|\tau) = 1$. We now check the converse, that when $K$ has control it is convex. Indeed, for any $\phi_0,\phi_1 \in K$ choose their control operation $\phi \in \textnormal{\texttt{dExt}}_{Y,I}(K)$. Consider an arbitrary convex combination $p\phi_0 + (1-p)\phi_1$, this combination is given by inserting $ \sigma := p\rho_0 + (1-p) \rho_1$ into the wire $Y$ of $\phi$. Since $\phi$ is in $\textnormal{\texttt{dExt}}_{Y,I}$ then insertion of $\sigma$ into $\phi$ must return an element of $K$ and so it follows that the convex combination $p\phi_0 + (1-p)\phi_1$ is an element of $K$.
\end{proof}
This equivalence is noted to furthermore hold for classical information theory in the appendix, by an identical proof method. We finish by stating a definition we will find convenient to reference later.
\begin{definition}
The operational closure $K_{\mathbf{CP}}$ of a set $K \subseteq \mathbf{QC}(A,A')$ is given by the set of all: \[  \input{figs/extension_3b.tikz} \] where $\phi \in \textnormal{\texttt{dExt}}_{X,X'}(K)$, $\rho \in \mathbf{CP}(I,X)$, and the effect $ \sigma \in  \mathbf{CP}(X',I)$ for some $X,X'$.
\end{definition}
The operational closure $K_{\mathbf{CP}}$ of $K$ represents the set of all operations onto which locally-applicable transformations on $K$ can be uniquely extended.

Finally, we will need one extra condition on convex sets, that of at least containing all discard-prepare operations. 
\begin{definition}
    A set $K\subseteq QC(A,A')$ is normal if for every effect $d: A \rightarrow I$ and state $\rho : I \rightarrow A'$ then the process $\rho \circ d : A \rightarrow A'$ is in $K$.
\end{definition}
The reason that such sets are of importance to us is that the swap morphism, i.e the monoidal symmetry, is guaranteed to be within the dilation extension.  
\begin{lemma}
    Let $K \subseteq \mathbf{QC}(A,A')$ be a normal convex set, then $\texttt{SWAP}_{A,A'} \in\textnormal{\texttt{dExt}}_{A',A}(K)$
\end{lemma}
\begin{proof}
    Every process of $\mathbf{QC}(A,A')$ given by applying a state and effect to half of the swap channel, is a discard-prepare channel.
\end{proof}

\subsection{Proof of main theorem}

In this section we prove that the principle of local-applicability is sufficient to characterize quantum supermaps. Concretely we show that for any locally-applicable transformation $S$ on $\mathbf{QC}$ of type $K \rightarrow M$ with $K,M$ normal and convex,  there exists a $\mathbf{CP}$-supermap $S_{Q}$ of type $K \rightarrow M$ which implements it. Note that we from now on identify a set $K \subseteq \mathbf{QC}(A,A')$ with its dilation extension $\textnormal{\texttt{dExt}}_{-,=}(K)$ for convenience, so that we may refer to a supermap of type $\textnormal{\texttt{dExt}}_{-,=}(K) \rightarrow \textnormal{\texttt{dExt}}_{-,=}(M)$ as simply a supermap of type $K \rightarrow M$. The formal meaning of implementation is given by the existence of some $S_Q$ such that $S = \mathcal{F}_{\mathbf{CP},\mathbf{QC}}({S}_{Q})$ where $\mathcal{F}_{\mathbf{CP},\mathbf{QC}}$ is the previously defined embedding from $\mathbf{CP}$-supermaps into locally-applicable transformations. We will note in passing that, as a consequence, we will have constructed an \textit{equivalence of categories} between quantum supermaps on signalling constraints and locally-applicable transformations on pathing constraints, where the latter definition is void of any reference to compact closure, linearity, coarse-graining, or even causal structure. The equivalence will in fact hold for all locally-applicable transformations $K \rightarrow M$ between normal convex sets $K,M$ whether or not they be associated to signalling constraints. We begin by deriving convex linearity from locality in a general setting.

\begin{lemma}[Convex linearity]
Let $K,M$ be convex, then every locally-applicable transformation of type $S: K \rightarrow M$ on $\mathbf{QC}$ preserves convex combinations.
\end{lemma}
\begin{proof}
Consider a pair $\phi_i \in \textnormal{\texttt{dExt}}_{X,X'}(K)$ of channels, since $K$ is convex it has control, there exists $\phi \in \textnormal{\texttt{dExt}}_{X,X'}(K)$ such that \[    \input{figs/control_1_caus.tikz} , \] and so an arbitrary convex combination $p \phi_o + (1-p) \phi_1$ can be written as \[  p \phi_o + (1-p) \phi_1 \quad = \quad    \input{figs/convex_linear_2.tikz} , \] with $\rho_p := p \rho_o + (1-p) \rho_1$. Writing $p_1 = (1-p)$, then $ S_{X,X'}(p_0 \phi_0 + p_1 \phi_1 )$ is given using naturality by: \[\input{figs/convex_linear_3.tikz} \quad = \quad  \input{figs/convex_linear_4.tikz}  \quad   = \quad  \input{figs/convex_linear_5.tikz} .  \] Rewriting $\rho_{p}$ in terms of $\rho_0,\rho_1$ gives
\[     = \quad  \input{figs/convex_linear_6a.tikz} \quad + \quad  \input{figs/convex_linear_6b.tikz}, \] then using naturality again \[  = \quad \input{figs/convex_linear_7a.tikz} \quad + \quad  \input{figs/convex_linear_7b.tikz} ,  \] and finally using the definition of control recovers the result \[  = \quad \input{figs/convex_linear_8a.tikz} \quad + \quad  \input{figs/convex_linear_8b.tikz}  \quad = \quad  p_0 S_{X,X'}(\phi_0) + p_1 S_{X,X'}(\phi_1). \] 
\end{proof}
We note that $S$ consequently has a unique extension to the real-linear span of $K$. We now show using this result that locally-applicable transformations on $\mathbf{QC}(A,A')$ can be uniquely extended to $\mathbf{CP}(A,A')$, more generally for any set $K$ the map $S$ can be extended to the operational closure $K_{\mathbf{CP}}$ of $K$.

\begin{lemma}[Extension to operational closure]
Let $K,M$ be convex, then every locally-applicable transformation of type $S: K \rightarrow M$ on $\mathbf{QC}$ has a unique extension to a function ${S}_{\mathbf{CP}}:K_{\mathbf{CP}} \rightarrow M_{\mathbf{CP}}$.
\end{lemma}
\begin{proof}
We give a candidate definition and then show that it is well-posed, consider some $\phi \in K_{\mathbf{CP}}$, for this $\phi$ there exists $\phi \in \textnormal{\texttt{dExt}}_{X,X'}(K)$ and $\rho,\sigma$ such that: \[  \input{figs/extension_1b.tikz} \quad = \quad \input{figs/extension_3b.tikz} . \] With respect to this choice we can define 
\[ \input{figs/final7.tikz} \quad := \quad \input{figs/extension_2.tikz} . \]
We now show that this definition is well-posed by showing that for any other choice of $\phi,\rho,\sigma$ the resulting outcome $S_{\mathbf{CP}}(\phi)$ would be the same. Indeed, let \[   \input{figs/extension_3b.tikz} \quad = \quad \input{figs/extension_3a.tikz}  \] then consider the process $\Sigma$ defined by \[   \input{figs/extension_4a.tikz} \quad = \quad \input{figs/extension_4b.tikz} \quad + \quad \input{figs/extension_4c.tikz}    \] where the output is taken to be at least $2$-dimensional. A process $\Sigma'$ can be defined similarly, note that both $\Sigma,\Sigma'$ are trace preserving and so are members of $\mathbf{QC}$, meaning that they can be slid along dotted wires. We now consider the result of applying them to $\phi,\phi'$: \[   \input{figs/extension_5a.tikz} \quad = \quad \input{figs/extension_5b.tikz} \quad + \quad \input{figs/extension_5c.tikz} \quad - \quad   \input{figs/extension_5d.tikz} , \]
this in turn implies that \[  \input{figs/extension_5a.tikz} \quad -  \quad \input{figs/extension_6a.tikz} \quad = \quad \input{figs/extension_5c.tikz} \quad - \quad   \input{figs/extension_6c.tikz} ,  \] and so then by linearity \begin{equation}   \input{figs/extension_7a.tikz} \quad -  \quad \input{figs/extension_7b.tikz} \quad = \quad \input{figs/extension_7c.tikz} \quad - \quad   \input{figs/extension_7d.tikz}  \end{equation}
We begin by showing that $\sigma,\sigma'$ can be safely pulled across dotted wires maintaining equality even though they are not quantum channels, Indeed consider the difference, and express in terms of $\Sigma,\Sigma'$: 
\[   \input{figs/extension_8a.tikz} \quad -  \quad \input{figs/extension_8b.tikz} \quad = \quad \input{figs/extension_9a.tikz} \quad - \quad   \input{figs/extension_9b.tikz},  \] 
Now, since $\Sigma,\Sigma'$ \textit{are} quantum channels they can be pulled through dotted lines, after which we use equation (1): 
\[ = \quad   \input{figs/extension_10a.tikz} \quad -  \quad \input{figs/extension_10b.tikz} \quad = \quad \input{figs/extension_11a.tikz} \quad - \quad   \input{figs/extension_11b.tikz} . \] 
The preparation $1:I \rightarrow X'$ is a quantum channel it can be pulled back through dotted wires, after which orthogonality implies that the difference has to be $0$:
\[ = \quad   \input{figs/extension_12a.tikz} \quad -  \quad \input{figs/extension_12b.tikz} \quad = \quad 0 . \] 
Since the difference is $0$ then it follows that \[  \input{figs/extension_8a.tikz} \quad =  \quad \input{figs/extension_8b.tikz} . \] We can now consider the bottom side, which is easier to reason with since every $\rho \in \mathbf{CP}(I,X)$ is given by $a \eta$ with $\eta$ a state in $\mathbf{QC}(I,X)$ and $a$ a scalar. Indeed, using the extension of $S$ to $\mathbb{R}$-linearity (and so to multiplication by scalar $\alpha$) gives \[  \input{figs/extension_8a.tikz} \quad =  \quad \input{figs/extension_8a1.tikz}  \quad =  \quad \input{figs/extension_8a3.tikz}  \quad  =  \quad \input{figs/extension_8a4.tikz},  \] and similarly for the $\phi',\rho',\sigma'$ we find \[  \input{figs/extension_8b.tikz} \quad =  \quad \input{figs/extension_8b4.tikz} .   \] All together then, using the seperate deductions made for $\rho$ and $\sigma$ gives  \[  \input{figs/extension_8b4.tikz} \quad =  \quad \input{figs/extension_8b.tikz}  \quad =  \quad \input{figs/extension_8a.tikz}  \quad =  \quad \input{figs/extension_8a4.tikz}  ,  \] and so $S_{\mathbf{CP}}$ is indeed well defined.
\end{proof}
The above is the key to our result, we now are ready to construct a candidate quantum supermap for simulating the action of our locally-applicable transformation $S_{X,X'}$ by tensor extension with identities on $X,X'$. To do so we apply our locally-applicable transformation $S$ to the swap-morphism, the intuition being that the swap gives a way to noiselessly extract information about the input behaviour of a higher order map, by converting its input into a pair of lower-order objects.
\begin{theorem}[Re-characterisation of supermaps]
Let $K,M$ be normal convex subsets of channels of $\mathbf{QC}$, there is a one-to-one correspondence between quantum supermaps of type $K \rightarrow M$ and locally-applicable transformations of the same type.
\end{theorem}
\begin{proof}
Given a locally-applicable transformation $S$ of type $K \rightarrow M$ on $\mathbf{QC}$ with $K \subseteq \mathbf{QC}(A,A')$ and $M \in \subseteq \mathbf{QC}(B,B')$ we define $S_{Q}: A^{*} \otimes A' \rightarrow B^{*} \otimes B'$ by: \[  \tikzfigscale{0.6}{figs/final1} \textrm{ } := \quad \input{figs/final2.tikz}. \] In other words we apply $S$ to the swap in $\mathbf{QC}$ and then embed into $\mathbf{CP}$ so that we may apply caps and cups, note that normality of $K,M$ is required here to ensure that the swap lives within their dilation extensions. We now consider the application of arbitrary states and effects $\rho, \sigma$ in $\mathbf{CP}$ to the auxiliary wires, and use the tensor seperation property: \[  \tikzfigscale{0.6}{figs/final3} \textrm{ } = \quad \input{figs/final4.tikz} \quad = \quad \input{figs/final5.tikz} . \] Using the well-posed definition of $S_{\mathbf{CP}}$ and then using compact closure to replace the cup and cap with the identity gives: \[ = \quad \input{figs/final6.tikz} \quad = \quad \input{figs/final7.tikz} \quad = \quad \input{figs/final8.tikz} . \] Since this is true for all $\rho, \sigma$ in $\mathbf{CP}$ it follows that: \[  \tikzfigscale{0.6}{figs/final9} \textrm{ } = \quad \input{figs/final10.tikz}, \] and so there indeed exists a quantum supermap of type $S_{Q}:K \rightarrow M ,$ such that $\mathcal{F}_{\mathbf{CP}}(S_{Q}) = S$ where $\mathcal{F}_{\mathbf{CP}}$ is the previously defined embedding of $\mathbf{CP}$-supermaps on $\mathbf{QC}$ into locally-applicable transformations on $\mathbf{QC}$.
\end{proof}
Note that in $\mathbf{QC}$ the set $\mathcal{E}_{sig}(\tau)$ is convex and that furthermore $\mathcal{E}_{sig}(\tau) = \mathcal{E}_{path}(\tau)$. It follows then that for arbitrary relations $\tau$ the locally-applicable transformations of type $\mathcal{E}_{path}(\tau) \rightarrow \mathcal{E}_{path}(\lambda)$ characterize the quantum supermaps of type $\mathcal{E}_{sig}(\tau) \rightarrow \mathcal{E}_{sig}(\lambda)$. As a stricter corollary the quantum supermaps on non-signalling channels, also referred to in the literature as process matrices have been characterised from principles of compositionality, without reference to causality or preferred time direction. To phrase this concisely we refer to the morphisms in: \[\mathcal{E}_{path}\Big(\input{figs/phi_nosig.tikz}\Big) , \]as non-pathing morphisms and refer to supermaps and locally-applicable transformations of type $K \rightarrow \mathbf{C}(B,B')$ for some $B,B'$ as being $\textit{on}$ the set $K$.
\begin{corollary}
There is a one-to-one correspondence between quantum supermaps on the set of non-signalling channels and locally-applicable transformations on the set of non-pathing channels.
\end{corollary}
Other example of convex sets of interest for which supermaps are now consequently characterised in terms of locality are those convex sets which are specified by sectorial constraints \cite{VanrietveldeRoutedCircuits,Wilson2021ComposableConstraints, Ormrod2022CausalSwitch, vanrietvelde2021coherent} which find use in the phrasing of more complex instances of pathing constraints \cite{Lorenz2020CausalTransformationsb} and in the analysis of fine-grained causal structures \cite{Ormrod2022CausalSwitch, Vilasini2022EmbeddingMatrices}. Finally, we comment that by denoting the restriction of the category $\mathbf{CPsup}[\mathbf{QC}]$ to convex sets by $ \mathbf{CPsup}[\mathbf{QC}]_{con}:= \mathbf{QS}_{con}$ and similarly denoting restriction of the category $\mathbf{lot}[\mathbf{C}]$ to convex sets by $\mathbf{lot}[\mathbf{C}]_{con}$ we can state the one-to-one correspondence in concise categorical language.
\begin{corollary}
There is an equivalence of categories $\mathbf{QS}_{con} \cong \mathbf{lot}[\mathbf{QC}]_{con}$
\end{corollary}
This equivalence is confirmed by noting that the assignment $\mathcal{F}_{\mathbf{CP}}$ is functorial, surjective on objects, and full and faithful by the one-to-one correspondence observed in this section. In summary, quantum supermaps and equivalently process matrices, originally defined in terms of Choi-isomorphisms and probabilistic structure respectively, are instances of a purely compositional definition of higher order mapping of quantum channels.

\section{Multiparty Approach}
In this section we generalise the definition of locally-applicable transformation to multi-input transformations. Whilst we have technically already recovered the standard definition supermaps on multiple parties using locally-applicable transformations on non-signalling processes, the following approach has the advantage that it comes closer to constructing for free (without characterization in the quantum setting) some key compositional features of quantum supermaps.
\begin{definition}
A locally-applicable transformation of type $K^1_{-,=} \dots K^{n}_{-,=} \longrightarrow M_{-,=}$ is a family of functions \[S_{X_1 \dots X_n}^{X_1' \dots X_n'}: K^1_{X_1,X_1'} \dots K^{n}_{X_n,X_n'} \longrightarrow M_{X_1 \dots X_n,X_1' \dots X_n'} \] satisfying: \[ \input{figs/multi_slot_1a.tikz} \quad = \quad \input{figs/multi_slot_1b.tikz} . \]
\end{definition}
This definition for supermaps with multiple inputs allows to more easily include important examples of supermaps on arbitrary symmetric monoidal categories, such as combs \cite{Chiribella2008QuantumArchitecture} on general monoidal categories \cite{Coecke2014AResources}.
\begin{example}
Let $\mathbf{C}$ be a symmetric monoidal category, the locally-applicable transformation $\texttt{comb}[c_1 \dots c_{n+1}]$ of type $\texttt{comb}[c_1 \dots c_{n+1}] : \mathbf{C}(A_1 - , A_1' =) \dots \mathbf{C}(A_n - , A_n' =) \rightarrow \mathbf{C}(B - , B =)$ is the family of functions of type $\texttt{comb}[c_1 \dots c_{n+1}]_{X_1 \dots X_n,X_1' \dots X_n'} : \mathbf{C}(A_1 X_1 , A_1' X_1') \dots \mathbf{C}(A_n X_n , A_n' X_n') \rightarrow \mathbf{C}(B X_1 \dots X_n , B X_1' \dots X_n' )$ given by \[ \tikzfigscale{1}{figs/rep_big} . \]
\end{example}
Again, this definition can be extended to those combs of type $\texttt{comb}(c_1 \dots c_{n+1}) : K^1_{-,=} \dots K^{n}_{-,=} \longrightarrow M_{-,=}$ which are those combs such that for all $\phi_i \in K^i_{X_i,X_i'}$ then $\texttt{comb}[c_1 \dots c_{n+1}]_{X_1 \dots X_n,X_1' \dots X_n'}(\phi_1 \dots \phi_n) \in M_{ X_1 \dots X_n  , X_1' \dots X_n' }$. This definition can further be generalised to combs in $\mathbf{D}$ for any symmetric monoidal category $\mathbf{D}$ into which $\mathbf{C}$ is included.
\begin{example}
Let $\mathbf{C} \subseteq \mathbf{D}$ be an inclusion of symmetric monoidal categories and $K^i_{-,=},M_{-,=}$ be extendable sets of $\mathbf{C}$, then the $\mathbf{D}$-combs of type $K^1_{-,=} \dots K^{n}_{-,=} \longrightarrow M_{-,=}$ on $\mathbf{C}$ are those $\texttt{comb}[d_1 \dots d_{n+1}]$ in $\mathbf{D}$ such that for any $\phi_i \in K^i_{X_i,X_i'}$ then $\texttt{comb}[d_1 \dots d_{n+1}]_{X_1 \dots X_n,X_1' \dots X_n'}(\phi_1 \dots \phi_n) \in M_{ X_1 \dots X_n  , X_1' \dots X_n' }$.
\end{example}

By a proof method identical to that which is given in the singly-party case, for any compact closed category $\mathbf{P}$ a one-to-one correspondence can be given between the multiparty $\mathbf{P}$-combs and the multiparty $\mathbf{P}$-supermaps on any symmetric monoidal category $\mathbf{C} \subseteq \mathbf{P}$.

The advantage of the multi-input approach is that it freely recovers two key aspects of the compositional semantics of supermaps, the first being that they may be composed via nesting. Diagrammatically this nesting composition is given by taking $ S \circ (T^1 \dots T^m)(\phi_i^j)$ to be  \[  \input{figs/multi_slot_comp1.tikz} . \]
In category theoretic terms this nesting composition means that locally-applicable transformations always define a multicategory, with objects given by extendable sets $K_{-,=}$ and multi-morphisms of type $K^1_{-,=} \dots K^{n}_{-,=}  \rightarrow M_{-,=}$ given by locally-applicable transformations of the same type. This generalises the multi-categorical structures inherited by monoidal structure of $\mathbf{Caus}[\mathbf{C}]$ \cite{Kissinger2019AStructure}. The second key compositional feature of such supermaps is their enriched structure \cite{wilson2022Aprocesses}, there always exists locally-applicable transformations with multiple inputs which simply compose their input processes in sequence, or in parallel. The former appears as a locally-applicable transformation of type $\circ:\mathbf{C}(A,B)\mathbf{C}(B,C) \rightarrow \mathbf{C}(A,C)$: \[\input{figs/enrich_1.tikz}, \] the latter appears as a locally-applicable transformation of type $\mathbf{C}(A,A') \mathbf{C}(B,B') \rightarrow \mathbf{C}(A \otimes B,A' \otimes B')$: \[\input{figs/enrich_2.tikz}. \] We note here that at this level enrichment is in a multicategory rather than the more standard setting of enrichment in a monoidal category \cite{wilson2022Aprocesses}.

Consequently, the multi-input approach has some structural advantages. Conveniently, the results on single-input locally-applicable transformations generalise to multi-input locally applicable transformations by noting that when all but one input is filled, what remains is a standard locally applicable transformation.

\begin{corollary}
let $K_1, \dots , K_n, M$ be convex sets of morphisms of $\mathbf{QC}$, there is a one-to-one correspondence between $\mathbf{CP}$-supermaps of type $K_1 \dots K_n \rightarrow M$ and locally-applicable transformations of the same type.
\end{corollary}
\begin{proof}
That $\mathbf{CP}$-supermaps still give locally-applicable transformations follows from multiple uses of the interchange law for symmetric monoidal categories. What remains is to prove that every locally-applicable transformation is implemented by a $\mathbf{CP}$-supermap. Up to braiding the family of functions given by $S((-)\psi_2 \dots \psi_n)_{X_1,X_1'}(\psi_1) := S(\psi_1 \dots \psi_n)_{X_1 \dots X_n, X_1' \dots X_n'}$ is a locally-applicable transformation with one-input, consequently we can use our main theorem to show that \[\input{figs/multinew1.tikz}\] is equal to \[ \input{figs/multinew2.tikz}. \] Repeating this step for each consecutive input from $2$ to $n$ returns  \[ \input{figs/multinew3.tikz}, \] which completes the proof. 
\end{proof}
The $\mathbf{CP}$-supermaps of type $\mathbf{C}(A_1,A_1')  \mathbf{C}(A_2,A_2') \rightarrow \mathbf{C}(B,B')$ are in one to one correspondence with $\mathbf{CP}$-supermaps of type \[\mathcal{E}_{path}\Big(\input{figs/phi_nosig.tikz}\Big)  \longrightarrow \mathbf{C}(B,B') \] which can be viewed as a consequence of linear distributivity of the $\mathbf{Caus}[\mathbf{C}]$ construction along with the fact that non-signalling channels are given by the double closure as defined in \cite{Kissinger2019AStructure} of the set of product channels. This statement can be generalized, for the same reasons the $\mathbf{CP}$-supermaps of type $\mathbf{C}(A_1,A_1') \dots \mathbf{C}(A_n,A_n') \rightarrow \mathbf{C}(B,B')$ are in one-to-one correspondence with the $\mathbf{CP}$-supermaps on the $n$-input/$n$-output non-signalling channels. 

Consequently, the locally-applicable transformations with multiple inputs give another way to characterize the supermaps on non-signalling channels used to study indefinite causal structure. This time, the construction provides for free some key compositional features of such supermaps.

\section{Summary and Outlook}

A definition of locally-applicable transformations is introduced which refers only to the circuit-theoretic structure of deterministic quantum information theory. By being purely compositional in nature it may be applied to arbitrary symmetric monoidal categories and stated concisely in the language of category theory using the notion of a natural transformation. When applied to quantum channels, including those equipped with signaling constraints, locally-applicable transformations are in one to one correspondence with quantum supermaps, so provide a re-axiomatisation for supermaps in terms of the principles of sequential composition, parallel composition, and locality.

A clear application of the redefinition of supermaps as locally applicable transformations is that it provides a candidate definition for supermaps on arbitrary operational probabilistic theories (OPTs) \cite{Chiribella2016QuantumPrinciples}. Indeed, the problem of defining indefinite causal orders over generalized physical theories has recently gained interest \cite{bavaresco2024indefinitecausalorderboxworld, sengupta2024achievingmaximalcausalindefiniteness}. The characterization theorems in this article (along with the outline of their analogues for classical information theory in the appendix) demonstrate that this definition is a particularly strong candidate for supermaps on OPTs, as it recovers the established definitions of supermaps in cases where they are well understood. Efficient application of this categorical approach would further benefit from the characterization of locally applicable transformations in broader classes of OPTs, as well as from addressing the problem of extracting probabilistic correlations from locally-applicable transformations (i.e those used to device-independently verify properties of causal structures \cite{Oreshkov2012QuantumOrder}).
 
Beyond this point, the construction presented in this article could in principle be a seed for a variety of new projects, including comparison of its infinite iteration with the construction of higher order causal categories, the free construction of various compositional features of quantum supermaps \cite{Kissinger2019AStructure, Bisio2019TheoreticalTheory, wilson2022Aprocesses}, and characterisation of locally-applicable transformations on infinite dimensional quantum theory. Less concretely and more broadly it is the authors' hope that this definition of supermap is closer to allowing for a suitable generalisation to more elaborate notions of space and parallel composition  \cite{Gogioso2019ASpace, Arrighi2021QuantumTheory, Moliner2017SpaceCategories} including but not limited to algebraic quantum field theories \cite{Haag1964AnTheory}, so that higher order quantum theories can be brought into closer contact with theories of quantum gravity where supermaps are hoped to provide insight by modelling of quantum causal structure \cite{Hardy2006TowardsStructure}.

\subsubsection*{Acknowledgements}
MW is grateful to J Hefford, C Comfort, A Vanreitvelde, N Ormrod, N Pinzani, and H Kristjánsson for useful conversations, and again to N Ormrod for reviewing an earlier draft of this manuscript. AK and GC are supported by the John Templeton Foundation through grant 62312, The Quantum Information Structure of Spacetime (qiss.fr). The opinions expressed in this publication are those of the authors and do not necessarily reflect the views of the John Templeton Foundation. GC was supported by the Croucher Foundation and by the Hong Kong Research Grant Council (RGC) though the Senior Research Fellowship Scheme SRFS2021-7S02. AK would like to acknowledge additional support from the Engineering and Physical Sciences Research Council grant number EP/Z002230/1: (De)constructing quantum software (DeQS).
MW was supported by University College London and the EPSRC Doctoral Training Centre for Delivering Quantum Technologies [grant numbers EP/L015242/1 and EP/W524335/1].

\bibliographystyle{utphys.bst}
\bibliography{ref_local}

\appendix

\section{Category Theory}
Here we introduce three key concepts in abstract algebra, categories, functors, and natural transformations. Avoiding issues of size, a category is a set of objects $A,B,\dots$ with for each pair $A,B$ a set $\mathbf{C}(A,B)$ of morphisms often denoted as $f:A \rightarrow B$. A category is furthermore equipped with, for each triple $(A,B,C)$ of objects a composition function $\circ: \mathbf{C}(A,B) \times \mathbf{C}(B,C) \rightarrow \mathbf{C}(A,B)$ which is associative so that $f \circ (g \circ h) = (f \circ g) \circ h$, and for each object $A$ a unit $id_A$  so that for all $f:A \rightarrow B$ then $f \circ id_A = f = id_B \circ f$. In graphical terms a morphism $f:A \rightarrow A'$ may be written as a box labelled $f$ with input wire $A$ and output wire $A'$, sequential ($\circ$) composition represented by:
\[ g \circ f \quad := \quad \tikzfigscale{1}{figs/process_1} \quad \quad \textrm{and} \quad  \quad id \quad = \quad \input{figs/wire.tikz} \] The identity box, can be diagrammatically represented as a wire. This notation has the convenient property that it absorbs the structural equations of a category. There is no graphical difference for instance between the diagram representing $(f \circ g) \circ h$ and the diagram representing $f \circ (g \circ h)$. A functor $\mathcal{F}: \mathbf{C} \rightarrow \mathbf{D}$ is a structure preserving map between categories, formally it is an assignment of an object $\mathcal{F}A$ to each object $A$ along with for each pair $A,B$ a function $\mathcal{F}_{AB}:\mathbf{C}(A,B) \rightarrow \mathbf{D}(\mathcal{F}A,\mathcal{F}B)$ which preserves composition in the sense that $\mathcal{F}(f \circ g) = \mathcal{F}(f) \circ \mathcal{F}(g)$ and $\mathcal{F}(i) = i$. Graphically a functor can be represented by a surrounding box, which satisfies box-merging and identity removal: \[ \input{figs/functorbox1.tikz} \quad = \quad \input{figs/functorbox2.tikz} \quad \quad \quad \input{figs/functorbox3.tikz} \quad = \quad \input{figs/wire.tikz}. \] A natural transformation $\eta: F \Rightarrow G$ is for every object $A$ a morphism $\eta_{A}:\mathcal{F}A \rightarrow \mathcal{G}A$ such that for every $f:A \rightarrow B$ then $\eta_B \circ \mathcal{F}(f) = \mathcal{G}(f) \circ \eta_A$. Graphically this reads as: \[ \input{figs/functorboxnat2.tikz} \quad = \quad \input{figs/functorboxnat1.tikz}. \] For any category $\mathbf{C}$ the (\textit{reverse}) opposite category $\mathbf{C}^{op}$ can be defined in which $ \mathbf{C}^{op}(A,B) := \mathbf{C}(B,A)$ with composition and identity inherited from $\mathbf{C}$ so that $f \circ_{op} g := g \circ f$. Finally for any categories $\mathbf{C},\mathbf{D}$ the product category $\mathbf{C} \times \mathbf{D}$ can be defined in which objects are given by pairs $(A,B)$ and morphisms given by pairs of morphisms $(f,g)$ and composition rule inherited from $\mathbf{C}$ and $\mathbf{D}$ as $(f,g) \circ (f',g') := (f' \circ f, g' \circ g)$.

An example of a natural transformation of interest in this paper is given intuitively by the following: a family of functions $\eta_{A,B}: \mathbf{C}(A,B) \rightarrow \mathbf{C}(A,B)$ which commute with pre and post-composition by arbitrary processes, meaning that for all $f:A' \rightarrow A $ and $ g: B \rightarrow B' $: \[ \eta_{A',B'}(g \circ \phi \circ f) = g \circ \eta_{A,B}(\phi) \circ f  \] In this paper we choose to represent such families graphically by introducing a function-box notation. We represent $\eta_{A,B}$ by \[\eta_{A,B} \quad \approx \quad \input{figs/box1.tikz} \quad \quad \quad \eta_{A,B}(\phi) \quad \approx \quad \input{figs/box2.tikz}  \] The dotted wires are used to suggest the commutation condition satisfied by the family $\eta_{A,B}$: \[  \input{figs/box3.tikz} \quad = \quad  \input{figs/box4.tikz}  \] Such a family of functions can be phrased as a natural transformation in the following way, we begin with definition of the ``hom functor" $\mathbf{C}(-,=): \mathbf{C}^{op} \times \mathbf{C} \rightarrow \mathbf{Set}$. The category $\mathbf{C}^{op} \times \mathbf{C}$ has as objects pairs $(A,B)$ of objects of $\mathbf{C}$, for morphisms $(A,B) \rightarrow (A',B')$ it has pairs $(f:A' \rightarrow A, g:B \rightarrow B')$. A functor of type $\mathbf{C}^{op} \times \mathbf{C} \rightarrow \mathbf{Set}$ must send each pair $(A,B)$ to a set, indeed the functor $\mathbf{C}(-,=)$ is defined as sending the object $(A,B)$ to the set $\mathbf{C}(A,B)$. A functor of type $\mathbf{C}^{op} \times \mathbf{C} \rightarrow \mathbf{Set}$ must send each pair $(f:A' \rightarrow A, g:B \rightarrow B')$ to a function  $\mathbf{C}(A,B) \rightarrow \mathbf{C}(A',B')$. Indeed the functor $\mathbf{C}(-,=)$ is defined by sending each pair $(f,g)$ to the function 
\begin{align*}
& \mathbf{C}(f,g):\mathbf{C}(A,B) \rightarrow \mathbf{C}(A',B') \\ 
& \mathbf{C}(f,g)(\phi) = g \circ \phi \circ f 
\end{align*} Indeed this is a functor since $\mathbf{C}(i,i)(\phi) = \phi$ and $\mathbf{C}(f \circ f', g' \circ g)( \phi ) = g' \circ g \circ \phi  \circ f \circ f' =  \mathbf{C}(f', g' ) (\mathbf{C}(f, g) (\phi)) = ( \mathbf{C}(f', g' ) \circ \mathbf{C}(f, g)) (\phi)$. A natural transformation of type $\eta: \mathbf{C}(-,=) \Rightarrow \mathbf{C}(-,=)$ is a family of functions $\eta_{A,B}: \mathbf{C}(A,B) \rightarrow \mathbf{C}(A,B)$ such that $\eta_{A,B} \circ \mathbf{C}(f,g) = \mathbf{C}(f,g) \circ \eta_{A,B}$, which when applied as functions reads $\forall \phi: \textrm{ }\eta_{A',B'}(g \circ \phi \circ f) = g \circ \eta_{A,B}(\phi) \circ f  $, our original commutativity condition.

\section{Phrasing of localizability as naturality}
The previously introduced hom-functor $\mathbf{C}(-,=)$ can be generalised in monoidal categories to a functor of type \[\mathbf{C}(A \otimes -,A' \otimes =): \mathbf{C}^{op} \times \mathbf{C} \longrightarrow \mathbf{SET}\] which assigns to each $(X,X')$ the object $\mathbf{C}(A \otimes X,A' \otimes X')$ and to each morphism $f:Y \rightarrow X$ and each morphism $g:X' \rightarrow Y'$ the function 
\begin{align*}
    & \mathbf{C}(A \otimes f,A' \otimes g): \mathbf{C}(A \otimes X,A' \otimes X') \rightarrow \mathbf{C}(A \otimes Y,A' \otimes Y') \\
     & \mathbf{C}(A \otimes f,A' \otimes g) :: \phi \mapsto (i \otimes g) \circ \phi \circ (i \otimes f)
\end{align*}
This functor can be further generalised to a functor $\textnormal{\texttt{dExt}}(K)(-,=)$, which we now define.
\begin{definition}[Extension functor]
For every $K \subseteq \mathbf{C}(A,A')$ in a symmetric monoidal category $\mathbf{C}$ one can define a functor $\textnormal{\texttt{dExt}}(K)(-,=):\mathbf{C}^{op} \times \mathbf{C}  \rightarrow \mathbf{Set}$ given by
\begin{itemize}
    \item $\textnormal{\texttt{dExt}}(K)(X,X') := \textnormal{\texttt{dExt}}_{X,X'}(K)$
    \item $\textnormal{\texttt{dExt}}(K)(f,g): \textnormal{\texttt{dExt}}_{X,X'}(K) \rightarrow \textnormal{\texttt{dExt}}_{Y,Y'}(K)$ defined by \[   \textnormal{\texttt{dExt}}(K)(f,g)(\phi) := \input{figs/functor_ext1.tikz}   \]
\end{itemize}
\end{definition}
The functor $\mathbf{C}(A \otimes -,A' \otimes =)$ can be defined as the special case given by $ \textnormal{\texttt{dExt}}(\mathbf{C}(A,A'))$. $\textnormal{\texttt{dExt}}$ is well defined, whenever $\phi \in \textnormal{\texttt{dExt}}_{X,X'}(S)$ then $\textnormal{\texttt{dExt}}(S)(f,g)(\phi) \in \textnormal{\texttt{dExt}}_{Y,Y'}(S) $ since for each $f,g$ and $\rho,\sigma$ then \[    \input{figs/functor_ext2.tikz} \quad = \quad  \input{figs/functor_ext3.tikz} \quad \in \quad  K \] since $\phi \in \textnormal{\texttt{dExt}}_{X,X'}(K)$. Furthermore the assignment $\textnormal{\texttt{dExt}}(S)(f,g)$ is functorial since \[  \textnormal{\texttt{dExt}}(S)(f',g')(\textnormal{\texttt{dExt}}(S)(f,g)(\phi)) \quad  = \quad  \input{figs/functor_ext4.tikz} \quad = \quad \textnormal{\texttt{dExt}}(S)(f \circ f',g' \circ g)(\phi) \]

The definition of a locally-applicable transformation can be rephrased as the definition of a natural transformation respect to this functor. A natural transformation $S: \mathbf{C}(A \otimes -, A' \otimes = ) \longrightarrow \mathbf{C}(B \otimes -, B' \otimes = )$ will be any family of functions $S_{XX'}$ making the following diagram commute for all $f,g$: \[ \begin{tikzcd}
{\mathbf{C}(A \otimes X,A' \otimes X')} \arrow[rr, "S_{XX'}"] \arrow[d, "{\mathbf{C}(A \otimes f,A' \otimes g)}"'] &  & {\mathbf{C}(B \otimes X,B' \otimes X')} \arrow[d, "{\mathbf{C}(B \otimes f,B' \otimes g)}"] \\
{\mathbf{C}(A \otimes Y,A' \otimes Y')} \arrow[rr, "S_{YY'}"']                                                     &  & {\mathbf{C}(B \otimes Y,B' \otimes Y')}                                                    
\end{tikzcd}  \] In other words such that $S_{Y,Y'} \circ \mathbf{C}(B \otimes f,B' \otimes g) = \mathbf{C}(B \otimes f,B' \otimes g) \circ S_{X,X'}$. Evaluated on processes $\phi$, this condition reads $S_{Y,Y'} ( \mathbf{C}(B \otimes f,B' \otimes g)(\phi)) = \mathbf{C}(B \otimes f,B' \otimes g)(  S_{X,X'}( \phi))$, which unpacking the definition of $ \mathbf{C}(B \otimes f,B' \otimes g)$ is precisely the \textit{sliding} rule which in the case of quantum theory entails the stricter notion of locally-applicable transformation. This observation extends to general subsets, so that a locally-applicable transformation of type $S:K \rightarrow M$ in $\mathbf{QC}$ is exactly a natural transformation of type $S_{X,X'}:\textnormal{\texttt{dExt}}_{X,X'}(K) \rightarrow \textnormal{\texttt{dExt}}_{X,X'}(M)$. Locally applicable transformations of type $K^1_{-,=} \dots K^n_{-,=} \rightarrow M_{-,=}$ on $\mathbf{QC}$ can similarly be phrased as natural transformations of type \[  K^1_{-,=} \times \dots \times K^n_{-,=} \rightarrow M_{-,=}   \] where for any $\mathcal{F}:\mathbf{C}_{1} \rightarrow \mathbf{C}_{2}$ and $\mathcal{G}:\mathbf{D}_{1} \rightarrow \mathbf{D}_{2}$ the product functor $\mathcal{F} \times \mathcal{G} : \mathbf{C}_{1} \times \mathbf{D}_{1} \rightarrow \mathbf{C}_{2} \times \mathbf{D}_{2}$ is defined by $\mathcal{F} \times \mathcal{G} (c,d) = (\mathcal{F}(c),\mathcal{G}(d))$ and similarly on morphisms. Noting that functors of type $\mathbf{C}^{op} \times \mathbf{C} \rightarrow \mathbf{SET}$ are those which are termed endo-profunctors on $\mathbf{C}$, we can conclude that the study of quantum supermaps is the study of natural transformations between certain endo-profunctors on the category of quantum channels. This connection is of particular interest since such profunctors are currently used in the field of applied category theory to build profunctor-optics, which are closely related to combs in arbitrary monoidal categories \cite{Chiribella2008QuantumArchitecture, Roman2020OpenCalculus, Roman2020CombFeedback, Boisseau2022CorneringOptics, Hedges2017CoherenceGames, Riley2018CategoriesOptics, Ghani2016CompositionalTheory, Pollock2015Non-MarkovianCharacterisation, Hefford2022CoendCombs}.

\section{$*$\textbf{Hilb} Supermaps}
We introduce functor box notation for weak symmetric monoidal functors. Whilst $\mathcal{F}(f)$ will be notated as before,
for a weak monoidal functor, functorality is only-up-to ismorphism so that we may write: \[   \input{figs/functorbox1.tikz} \quad \cong \quad   \input{figs/functorbox2.tikz}  \]
However, we will say that a functor is $2$-faithful if \[   \mathcal{F}(f) \cong \mathcal{F}(g) \quad \implies \quad \input{figs/infinitefunct5.tikz} \quad = \quad   \input{figs/infinitefunct6.tikz} \]
The above allows us to generalise $\mathbf{D}$-representable supermaps to a setting which allows us to use compact closure when defining supermaps on infinite dimensional quantum systems. \begin{lemma}
Let $\mathbf{C}$ be a symmetric monoidal $2$-category with trivial $2$-morphisms and $\mathbf{D}$ be a symmetric monoidal $2$-category with a weak $2$-faithful symmetric monoidal $2$-functor $\mathcal{G}:\mathbf{C} \rightarrow \mathbf{D}$. Any $\texttt{comb}[a,b]$ with $a,b \in \mathbf{D}$ such that for all $\phi \in \textnormal{\texttt{dExt}}_{X,X'}(K)$ there exists $\psi \in \textnormal{\texttt{dExt}}_{X,X'}(M)$ such that $\mathcal{G}(\psi) = \texttt{comb}[a,b]_{\mathcal{G}X\mathcal{G}X'}(\phi)$ defines a locally-applicable transformation by taking $S_{X,X'}(\phi) $ to be the unique $\psi$ s.t $\mathcal{G}(\psi) = 
\texttt{comb}[a,b](\phi)$. Such a locally-applicable transformation is termed a  $\mathcal{G}$~representable supermap on $\mathbf{C}$ of type $K \rightarrow M$.
\end{lemma}
\begin{proof}
Note that \[   \input{figs/infiniteproofbegin.tikz} \quad = \quad   \input{figs/infiniteproof2.tikz}   \] \[    \cong \quad \input{figs/infiniteproof3.tikz} \quad = \quad   \input{figs/infiniteproof4.tikz} \] and so \[   \input{figs/infiniteproofbegin.tikz} \quad \cong \quad   \input{figs/infiniteproofend.tikz}   \] which by $2$-faithful-ness of $\mathcal{F}$ gives \[ \input{figs/formal_cp_1aa.tikz} \quad = \quad \input{figs/fornal_cp_1bb.tikz}.  \]
\end{proof}
This gives a way to construct and represent examples of supermaps on the category of unitaries $\mathbf{sepU} \subseteq \mathbf{sephilb}$ by using the embedding of $\mathbf{sephilb}$ into $*\mathbf{Hilb}$ \cite{Gogioso2017Infinite-dimensionalMechanics}.
\begin{example}
There is a $2$-faithful weak symmetric monoidal $2$-functor $\mathcal{G}:\mathbf{sepU} \rightarrow {}^{*}\mathbf{Hilb}$ given by composition of the embedding $\mathbf{sepU} \subseteq \mathbf{sephilb}$ and the truncation functor $\texttt{trunc}[-]_{w}:\mathbf{sephilb} \rightarrow {}^{*}\mathbf{Hilb}$ \cite{Gogioso2017Infinite-dimensionalMechanics}. The induced supermaps are then termed $\texttt{trunc}[-]$~representable supermaps on $\mathbf{C}$.
\end{example}
One can straight-forwardly generalise the above construction to define $\mathcal{G}$-representable supermaps with multiple inputs, and so in particular define $\texttt{trunc}[-]$~representable supermaps on $\mathbf{sepU}$ with multiple inputs to realise the structural maps of monoidal enrichment, and furthermore infinite dimensional switches.


\section{Characterisation of Classical Supermaps}
The proof methods presented in the main text can also be used to characterize the supermaps on finite dimensional classical information theory. These supermaps were referred to in \cite{Baumeler2015TheOrder} as classical processes, and by coinciding with the effects on the $n$-input/$n$-output non-signalling channels in \cite{Kissinger2019AStructure} can be simply defined as the $\mathbf{Mat}[\mathbb{R}_+]$-supermaps on $\mathbf{Stoch}$. We outline the story here, noting the key features common to quantum classical information theory that were used in our proof. First we must declare what would be aimed to be proven, and what we mean by classical information theory. We define deterministic classical maps to be stochastic matrices, just as with quantum channels it is easier to first state a compact closed category from which they are constructed.
\begin{definition}[Positive real matrices]
The category $\mathbf{Mat}[\mathbb{R}_{+}]$ of positive real matrices is given by taking as objects the positive integers $\mathbb{Z}_{+}$ and as morphisms of type $f:n \rightarrow m$ the matrices of dimension $n \times m$. Sequential composition is given by matrix multiplication and identity morphism is given by the diagonal matrix of ones.
\end{definition}
The category $\mathbf{Mat}[\mathbb{R}_{+}]$ is symmetric monoidal and furthermore compact closed.
\begin{definition}[Stochastic maps]
The category $\mathbf{Stoch}$ of stochastic maps is given by the subcategory of $\mathbf{Mat}[\mathbb{R}_{+}]$ which contains only those matrices with column vectors which sum to $1$.
\end{definition}
The category $\mathbf{Stoch}$ is symmetric monoidal and furthermore causal, the unique effect of type $n \rightarrow 1$ is given by the row vector $(1, \dots ,1)$. Using the inclusion $\mathbf{Stoch} \subseteq \mathbf{Mat}[\mathbb{R}]$ one can immediately define the $\mathbf{Mat}[\mathbb{R}]$-supermaps on $\mathbf{Stoch}$. We will now observe that these are precisely the locally-applicable transformations on $\mathbf{Stoch}$. First, we adress the equivalence between convexity and control. 
\begin{lemma}
A subset $K \in \mathbf{Stoch}(A,A')$ is convex if and only if it has control
\end{lemma}
\begin{proof}
All that was required to construct the proof was the existence of an object $Y$ with a pair of distinguishable states in the sense that $e_i \circ \rho_j = \delta_{ij}$ and the possibility to take positive sums. Sums are taken care of by $\mathbf{Mat}[\mathbb{R}_{+}]$ and $Y$ may be taken to be $2$. Indeed one can define $\rho_{i}:1 \rightarrow 2$ by taking the $k^{th}$ component of the column vector $\rho_{i}$ to be $\delta_{ik}$ and similarly for the effects $e_{j}:2 \rightarrow 1$.
\end{proof}
\begin{lemma}
Let $K \in \mathbf{Stoch}(A,A')$ and $M \in \mathbf{Stoch}(B,B')$ be convex subsets, then every locally-applicable transformation $S:K \rightarrow M$ is convex linear.
\end{lemma}
\begin{proof}
Follows directly from equivalence between convexity and control.
\end{proof}
We take the operational closure $K_{ \mathbf{Mat}[\mathbb{R}_{+}]}$ of a set $K \subseteq \mathbf{Mat}[\mathbb{R}_{+}]$ to be defined in the equivalence way as for subsets $M$ of $\mathbf{QC}$, replacing the applying of effects from $\mathbf{CP}$ with the applying of effects from $\mathbf{Mat}[\mathbb{R}_{+}]$ 
\begin{lemma}
Every locally-applicable transformation $S:K \rightarrow M$ between convex sets in $\mathbf{Stoch}$ extends to a function $S_{\mathbf{Mat}[\mathbb{R}_{+}]}: K_{\mathbf{Mat}[\mathbb{R}_{+}]} \rightarrow \mathbf{M}_{\mathbf{Mat}[\mathbb{R}_{+}]}$.
\end{lemma}
\begin{proof}
The required elements of the proof are that $\mathbf{Mat}[\mathbb{R}_{+}]$ embeds into $\mathbf{Mat}[\mathbb{R}]$ so that subtractions can be defined, and that for every effect $\sigma \in \mathbf{Mat}[\mathbb{R}_{+}](A,I)$ there exists $\lambda \in \mathbb{R}_{+}$ and $\sigma' \in \mathbf{Mat}[\mathbb{R}_{+}](A,I)$ such that $\lambda \sigma + \sigma'  =  (1, \dots ,1)$ (The discard).
\end{proof}
\begin{lemma}[Tensor seperation]
Every locally-applicable transformation of type $S:K \rightarrow M$ on $\mathbf{Stoch}$ tensor seperates. 
\end{lemma}
\begin{proof}
Follows since $\mathbf{Mat}[\mathbb{R}_{+}]$ has enough causal states.
\end{proof}
\begin{theorem}
For $K,M$ convex in $\mathbf{Stoch}$ there is a one-to-one correspondence between the $\mathbf{Mat}[\mathbb{R}_{+}]$-supermaps of type $K \rightarrow M$ on $\mathbf{Stoch}$ and the locally-applicable transformations of type $K \rightarrow M$ on $\mathbf{Stoch}$.
\end{theorem}
\begin{proof}
Here all that is required is compact closure of $\mathbf{Mat}[\mathbb{R}_{+}]$, and again that $\mathbf{Mat}[\mathbb{R}_{+}]$ has enough causal states.
\end{proof}
Again as a corollary of this theorem, the classical supermaps on non-signalling channels are characterised as locally-applicable transformations on non-pathing channels in $\mathbf{Stoch}$.
\end{document}

%% file: tikzstyles.tex
\usetikzlibrary{shapes.multipart}

\tikzstyle{bwSpider}=[
       rectangle split,
       rectangle split parts=2,
       rectangle split part fill={black,white},
 minimum size=3.6 mm, inner sep=-2mm, draw=black,scale=0.5,rounded corners=0.8 mm
       ]
 \tikzstyle{wbSpider}=[
       rectangle split,
       rectangle split parts=2,
       rectangle split part fill={white,black},
 minimum size=3.6 mm, inner sep=-2mm, draw=black,scale=0.5,rounded corners=0.8 mm
       ]
\tikzstyle{cWire}=[densely dotted, thick]
\tikzstyle{env}=[copoint,regular polygon rotate=0,minimum width=0.2cm, fill=black]

\tikzstyle{probs}=[shape=semicircle,fill=white,draw=black,shape border rotate=180,minimum width=1.2cm]

\tikzstyle{every picture}=[baseline=-0.25em,scale=0.5]
\tikzstyle{dotpic}=[] 
\tikzstyle{diredges}=[every to/.style={diredge}]
\tikzstyle{math matrix}=[matrix of math nodes,left delimiter=(,right delimiter=),inner sep=2pt,column sep=1em,row sep=0.5em,nodes={inner sep=0pt},text height=1.5ex, text depth=0.25ex]

\tikzstyle{inline text}=[text height=1.5ex, text depth=0.25ex,yshift=0.5mm]
\tikzstyle{label}=[font=\footnotesize,text height=1.5ex, text depth=0.25ex,yshift=0.5mm]
\tikzstyle{left label}=[label,anchor=east,xshift=1.5mm]
\tikzstyle{right label}=[label,anchor=west,xshift=-1.5mm]

\tikzstyle{braceedge}=[decorate,decoration={brace,amplitude=2mm,raise=-1mm}]
\tikzstyle{small braceedge}=[decorate,decoration={brace,amplitude=1mm,raise=-1mm}]

\tikzstyle{doubled}=[line width=1.6pt] 
\tikzstyle{boldedge}=[doubled,shorten <=-0.17mm,shorten >=-0.17mm]
\tikzstyle{boldedgegray}=[doubled,gray,shorten <=-0.17mm,shorten >=-0.17mm]
\tikzstyle{singleedgegray}=[gray]

\tikzstyle{semidoubled}=[line width=1.4pt] 
\tikzstyle{semiboldedgegray}=[semidoubled,gray,shorten <=-0.17mm,shorten >=-0.17mm]

\tikzstyle{boxedge}=[semiboldedgegray]

\tikzstyle{dottededge}=[dashed,shorten <=-0.02mm,shorten >=-0.02mm]

\tikzstyle{boldedgedashed}=[very thick,dashed,shorten <=-0.17mm,shorten >=-0.17mm]
\tikzstyle{vboldedgedashed}=[doubled,dashed,shorten <=-0.17mm,shorten >=-0.17mm]
\tikzstyle{left hook arrow}=[left hook-latex]
\tikzstyle{right hook arrow}=[right hook-latex]
\tikzstyle{sembracket}=[line width=0.5pt,shorten <=-0.07mm,shorten >=-0.07mm]

\tikzstyle{causal edge}=[->,thick,gray]
\tikzstyle{causal nondir}=[thick,gray]
\tikzstyle{timeline}=[thick,gray, dashed]

\tikzstyle{cedge}=[<->,thick,gray!70!white]

\tikzstyle{empty diagram}=[draw=gray!40!white,dashed,shape=rectangle,minimum width=1cm,minimum height=1cm]
\tikzstyle{empty diagram small}=[draw=gray!50!white,dashed,shape=rectangle,minimum width=0.6cm,minimum height=0.5cm]

\tikzstyle{dot}=[inner sep=0mm,minimum width=2mm,minimum height=2mm,draw,shape=circle]

\tikzstyle{phase dot}=[pdot,phase dimensions]
\tikzstyle{wphase dot}=[dot, phase dimensions]

\tikzstyle{leak}=[white dot, shape=regular polygon, minimum size=300mm, regular polygon sides=3, outer sep=-0.2mm, regular polygon rotate=270]
\tikzstyle{preleak}=[trapezium, trapezium angle=67.5, draw, inner sep=0.1pt, outer sep=0pt, minimum height=2mm, minimum width=4pt,rotate=270]
\tikzstyle{proj}=[white dot, shape=regular polygon, minimum size=3.3 mm, regular polygon sides=4, outer sep=-0.2mm]
\tikzstyle{Vleak}=[white dot, shape=regular polygon, minimum size=3.3 mm, regular polygon sides=3, outer sep=-0.2mm, regular polygon rotate=90]
\tikzstyle{dleak}=[white dot, line width=1.6pt, shape=regular polygon, minimum size=3.3 mm, regular polygon sides=3, outer sep=-0.2mm, regular polygon rotate=270]

\tikzstyle{Wsquare}=[white dot, shape=regular polygon, rounded corners=0.8 mm, minimum size=3.3 mm, regular polygon sides=3, outer sep=-0.2mm]
\tikzstyle{Wsquareadj}=[white dot, shape=regular polygon, rounded corners=0.8 mm, minimum size=3.3 mm, regular polygon sides=3, outer sep=-0.2mm, regular polygon rotate=180]
\tikzstyle{ddot}=[inner sep=0mm, doubled, minimum width=2.5mm,minimum height=2.5mm,draw,shape=circle]

\tikzstyle{black dot}=[dot,fill=black]
\tikzstyle{white dot}=[dot,fill=white,,text depth=-0.2mm]
\tikzstyle{white Wsquare}=[Wsquare,fill=gray,,text depth=-0.2mm]
\tikzstyle{white Wsquareadj}=[Wsquareadj,fill=white,,text depth=-0.2mm]
\tikzstyle{green dot}=[white dot] 
\tikzstyle{gray dot}=[dot,fill=gray!40!white,,text depth=-0.2mm]
\tikzstyle{red dot}=[gray dot] 

\tikzstyle{black ddot}=[ddot,fill=black]
\tikzstyle{white ddot}=[ddot,fill=white]
\tikzstyle{gray ddot}=[ddot,fill=gray!40!white]

\tikzstyle{gray edge}=[gray!60!white]

\tikzstyle{small dot}=[inner sep=0.5mm,minimum width=0pt,minimum height=0pt,draw,shape=circle]

\tikzstyle{small black dot}=[small dot,fill=black]
\tikzstyle{small white dot}=[small dot,fill=white]
\tikzstyle{small gray dot}=[small dot,fill=gray!40!white]

\tikzstyle{causal dot}=[inner sep=0.4mm,minimum width=0pt,minimum height=0pt,draw=white,shape=circle,fill=gray!40!white]

\tikzstyle{phase dimensions}=[minimum size=5mm,font=\footnotesize,rectangle,rounded corners=2.5mm,inner sep=0.2mm,outer sep=-2mm]
\tikzstyle{dphase dimensions}=[minimum size=5mm,font=\footnotesize,rectangle,rounded corners=2.5mm,inner sep=0.2mm,outer sep=-2mm]

\tikzstyle{white phase dot}=[dot,fill=white,phase dimensions]
\tikzstyle{white phase ddot}=[ddot,fill=white,dphase dimensions]

\tikzstyle{white rect ddot}=[draw=black,fill=white,doubled,minimum size=5mm,font=\footnotesize,rectangle,rounded corners=2.5mm,inner sep=0.2mm]
\tikzstyle{gray rect ddot}=[draw=black,fill=gray!40!white,doubled,minimum size=6mm,font=\footnotesize,rectangle,rounded corners=3mm]

\tikzstyle{gray phase dot}=[dot,fill=gray!40!white,phase dimensions]
\tikzstyle{gray phase ddot}=[ddot,fill=gray!40!white,dphase dimensions]
\tikzstyle{grey phase dot}=[gray phase dot]
\tikzstyle{grey phase ddot}=[gray phase ddot]

\tikzstyle{small phase dimensions}=[minimum size=4mm,font=\tiny,rectangle,rounded corners=2mm,inner sep=0.2mm,outer sep=-2mm]
\tikzstyle{small dphase dimensions}=[minimum size=4mm,font=\tiny,rectangle,rounded corners=2mm,inner sep=0.2mm,outer sep=-2mm]

\tikzstyle{small gray phase dot}=[dot,fill=gray!40!white,small phase dimensions]
\tikzstyle{small gray phase ddot}=[ddot,fill=gray!40!white,small dphase dimensions]

\tikzstyle{small map}=[draw,shape=rectangle,minimum height=4mm,minimum width=4mm,fill=white]

\tikzstyle{cnot}=[fill=white,shape=circle,inner sep=-1.4pt]

\tikzstyle{asym hadamard}=[fill=white,draw,shape=NEbox,inner sep=0.6mm,font=\footnotesize,minimum height=4mm]
\tikzstyle{asym hadamard conj}=[fill=white,draw,shape=NWbox,inner sep=0.6mm,font=\footnotesize,minimum height=4mm]
\tikzstyle{asym hadamard dag}=[fill=white,draw,shape=SEbox,inner sep=0.6mm,font=\footnotesize,minimum height=4mm]

\tikzstyle{hadamard}=[fill=white,draw,inner sep=0.6mm,font=\footnotesize,minimum height=4mm,minimum width=4mm]
\tikzstyle{small hadamard}=[fill=white,draw,inner sep=0.6mm,minimum height=1.5mm,minimum width=1.5mm]
\tikzstyle{small hadamard rotate}=[small hadamard,rotate=45]
\tikzstyle{dhadamard}=[hadamard,doubled]
\tikzstyle{small dhadamard}=[small hadamard,doubled]
\tikzstyle{small dhadamard rotate}=[small hadamard rotate,doubled]
\tikzstyle{antipode}=[white dot,inner sep=0.3mm,font=\footnotesize]

\tikzstyle{scalar}=[diamond,draw,inner sep=0.5pt,font=\small]
\tikzstyle{dscalar}=[diamond,doubled, draw,inner sep=0.5pt,font=\small]

\tikzstyle{small box}=[rectangle,inline text,fill=white,draw,minimum height=5mm,yshift=-0.5mm,minimum width=5mm,font=\small]
\tikzstyle{small gray box}=[small box,fill=gray!30]
\tikzstyle{medium box}=[rectangle,inline text,fill=white,draw,minimum height=5mm,yshift=-0.5mm,minimum width=10mm,font=\small]
\tikzstyle{square box}=[small box] 
\tikzstyle{medium gray box}=[small box,fill=gray!30]
\tikzstyle{semilarge box}=[rectangle,inline text,fill=white,draw,minimum height=5mm,yshift=-0.5mm,minimum width=12.5mm,font=\small]
\tikzstyle{large box}=[rectangle,inline text,fill=white,draw,minimum height=5mm,yshift=-0.5mm,minimum width=15mm,font=\small]
\tikzstyle{large gray box}=[small box,fill=gray!30]

\tikzstyle{Bayes box}=[rectangle,fill=black,draw, minimum height=3mm, minimum width=3mm]

\tikzstyle{gray square point}=[small box,fill=gray!50]

\tikzstyle{dphase box white}=[dhadamard]
\tikzstyle{dphase box gray}=[dhadamard,fill=gray!50!white]
\tikzstyle{phase box white}=[hadamard]
\tikzstyle{phase box gray}=[hadamard,fill=gray!50!white]

\tikzstyle{point}=[regular polygon,regular polygon sides=3,draw,scale=0.75,inner sep=-0.5pt,minimum width=9mm,fill=white,regular polygon rotate=180]
\tikzstyle{point nosep}=[regular polygon,regular polygon sides=3,draw,scale=0.75,inner sep=-2pt,minimum width=9mm,fill=white,regular polygon rotate=180]
\tikzstyle{copoint}=[regular polygon,regular polygon sides=3,draw,scale=0.75,inner sep=-0.5pt,minimum width=9mm,fill=white]
\tikzstyle{dpoint}=[point,doubled]
\tikzstyle{dcopoint}=[copoint,doubled]

\tikzstyle{pointgrow}=[shape=cornerpoint,kpoint common,scale=0.75,inner sep=3pt]
\tikzstyle{pointgrow dag}=[shape=cornercopoint,kpoint common,scale=0.75,inner sep=3pt]

\tikzstyle{wide copoint}=[fill=white,draw,shape=isosceles triangle,shape border rotate=90,isosceles triangle stretches=true,inner sep=0pt,minimum width=1.5cm,minimum height=6.12mm]
\tikzstyle{wide point}=[fill=white,draw,shape=isosceles triangle,shape border rotate=-90,isosceles triangle stretches=true,inner sep=0pt,minimum width=1.5cm,minimum height=6.12mm,yshift=-0.0mm]
\tikzstyle{wide point plus}=[fill=white,draw,shape=isosceles triangle,shape border rotate=-90,isosceles triangle stretches=true,inner sep=0pt,minimum width=1.74cm,minimum height=7mm,yshift=-0.0mm]

\tikzstyle{wide dpoint}=[fill=white,doubled,draw,shape=isosceles triangle,shape border rotate=-90,isosceles triangle stretches=true,inner sep=0pt,minimum width=1.5cm,minimum height=6.12mm,yshift=-0.0mm]

\tikzstyle{tinypoint}=[regular polygon,regular polygon sides=3,draw,scale=0.55,inner sep=-0.15pt,minimum width=6mm,fill=white,regular polygon rotate=180]

\tikzstyle{white point}=[point]
\tikzstyle{white dpoint}=[dpoint]
\tikzstyle{green point}=[white point] 
\tikzstyle{white copoint}=[copoint]
\tikzstyle{gray point}=[point,fill=gray!40!white]
\tikzstyle{gray dpoint}=[gray point,doubled]
\tikzstyle{red point}=[gray point] 
\tikzstyle{gray copoint}=[copoint,fill=gray!40!white]
\tikzstyle{gray dcopoint}=[gray copoint,doubled]

\tikzstyle{white point guide}=[regular polygon,regular polygon sides=3,font=\scriptsize,draw,scale=0.65,inner sep=-0.5pt,minimum width=9mm,fill=white,regular polygon rotate=180]

\tikzstyle{black point}=[point,fill=black,font=\color{white}]
\tikzstyle{black copoint}=[copoint,fill=black,font=\color{white}]

\tikzstyle{tiny gray point}=[tinypoint,fill=gray!40!white]

\tikzstyle{diredge}=[->]
\tikzstyle{ddiredge}=[<->]
\tikzstyle{rdiredge}=[<-]
\tikzstyle{thickdiredge}=[->, very thick]
\tikzstyle{pointer edge}=[->,very thick,gray]
\tikzstyle{pointer edge part}=[very thick,gray]
\tikzstyle{dashed edge}=[dashed]
\tikzstyle{thick dashed edge}=[very thick,dashed]
\tikzstyle{thick gray dashed edge}=[thick dashed edge,gray!40]
\tikzstyle{thick map edge}=[very thick,|->]

\makeatletter
\newcommand{\boxshape}[3]{%
\pgfdeclareshape{#1}{
\inheritsavedanchors[from=rectangle] 
\inheritanchorborder[from=rectangle]
\inheritanchor[from=rectangle]{center}
\inheritanchor[from=rectangle]{north}
\inheritanchor[from=rectangle]{south}
\inheritanchor[from=rectangle]{west}
\inheritanchor[from=rectangle]{east}
\backgroundpath{
\southwest \pgf@xa=\pgf@x \pgf@ya=\pgf@y
\northeast \pgf@xb=\pgf@x \pgf@yb=\pgf@y

\@tempdima=#2
\@tempdimb=#3

\pgfpathmoveto{\pgfpoint{\pgf@xa - 5pt + \@tempdima}{\pgf@ya}}
\pgfpathlineto{\pgfpoint{\pgf@xa - 5pt - \@tempdima}{\pgf@yb}}
\pgfpathlineto{\pgfpoint{\pgf@xb + 5pt + \@tempdimb}{\pgf@yb}}
\pgfpathlineto{\pgfpoint{\pgf@xb + 5pt - \@tempdimb}{\pgf@ya}}
\pgfpathlineto{\pgfpoint{\pgf@xa - 5pt + \@tempdima}{\pgf@ya}}
\pgfpathclose
}
}}

\boxshape{NEbox}{0pt}{5pt}
\boxshape{SEbox}{0pt}{-5pt}
\boxshape{NWbox}{5pt}{0pt}
\boxshape{SWbox}{-5pt}{0pt}
\boxshape{EBox}{-3pt}{3pt}
\boxshape{WBox}{3pt}{-3pt}
\makeatother

\tikzstyle{cloud}=[shape=cloud,draw,minimum width=1.5cm,minimum height=1.5cm]

\tikzstyle{map}=[draw,shape=NEbox,inner sep=2pt,minimum height=6mm,fill=white]
\tikzstyle{dashedmap}=[draw,dashed,shape=NEbox,inner sep=2pt,minimum height=6mm,fill=white]
\tikzstyle{mapdag}=[draw,shape=SEbox,inner sep=2pt,minimum height=6mm,fill=white]
\tikzstyle{mapadj}=[draw,shape=SEbox,inner sep=2pt,minimum height=6mm,fill=white]
\tikzstyle{maptrans}=[draw,shape=SWbox,inner sep=2pt,minimum height=6mm,fill=white]
\tikzstyle{mapconj}=[draw,shape=NWbox,inner sep=2pt,minimum height=6mm,fill=white]

\tikzstyle{medium map}=[draw,shape=NEbox,inner sep=2pt,minimum height=6mm,fill=white,minimum width=7mm]
\tikzstyle{medium map dag}=[draw,shape=SEbox,inner sep=2pt,minimum height=6mm,fill=white,minimum width=7mm]
\tikzstyle{medium map adj}=[draw,shape=SEbox,inner sep=2pt,minimum height=6mm,fill=white,minimum width=7mm]
\tikzstyle{medium map trans}=[draw,shape=SWbox,inner sep=2pt,minimum height=6mm,fill=white,minimum width=7mm]
\tikzstyle{medium map conj}=[draw,shape=NWbox,inner sep=2pt,minimum height=6mm,fill=white,minimum width=7mm]
\tikzstyle{semilarge map}=[draw,shape=NEbox,inner sep=2pt,minimum height=6mm,fill=white,minimum width=9.5mm]
\tikzstyle{semilarge map trans}=[draw,shape=SWbox,inner sep=2pt,minimum height=6mm,fill=white,minimum width=9.5mm]
\tikzstyle{semilarge map adj}=[draw,shape=SEbox,inner sep=2pt,minimum height=6mm,fill=white,minimum width=9.5mm]
\tikzstyle{semilarge map dag}=[draw,shape=SEbox,inner sep=2pt,minimum height=6mm,fill=white,minimum width=9.5mm]
\tikzstyle{semilarge map conj}=[draw,shape=NWbox,inner sep=2pt,minimum height=6mm,fill=white,minimum width=9.5mm]
\tikzstyle{large map}=[draw,shape=NEbox,inner sep=2pt,minimum height=6mm,fill=white,minimum width=12mm]
\tikzstyle{large map conj}=[draw,shape=NWbox,inner sep=2pt,minimum height=6mm,fill=white,minimum width=12mm]
\tikzstyle{very large map}=[draw,shape=NEbox,inner sep=2pt,minimum height=6mm,fill=white,minimum width=17mm]

\tikzstyle{medium dmap}=[draw,doubled,shape=NEbox,inner sep=2pt,minimum height=6mm,fill=white,minimum width=7mm]
\tikzstyle{medium dmap dag}=[draw,doubled,shape=SEbox,inner sep=2pt,minimum height=6mm,fill=white,minimum width=7mm]
\tikzstyle{medium dmap adj}=[draw,doubled,shape=SEbox,inner sep=2pt,minimum height=6mm,fill=white,minimum width=7mm]
\tikzstyle{medium dmap trans}=[draw,doubled,shape=SWbox,inner sep=2pt,minimum height=6mm,fill=white,minimum width=7mm]
\tikzstyle{medium dmap conj}=[draw,doubled,shape=NWbox,inner sep=2pt,minimum height=6mm,fill=white,minimum width=7mm]
\tikzstyle{semilarge dmap}=[draw,doubled,shape=NEbox,inner sep=2pt,minimum height=6mm,fill=white,minimum width=9.5mm]
\tikzstyle{semilarge dmap trans}=[draw,doubled,shape=SWbox,inner sep=2pt,minimum height=6mm,fill=white,minimum width=9.5mm]
\tikzstyle{semilarge dmap adj}=[draw,doubled,shape=SEbox,inner sep=2pt,minimum height=6mm,fill=white,minimum width=9.5mm]
\tikzstyle{semilarge dmap dag}=[draw,doubled,shape=SEbox,inner sep=2pt,minimum height=6mm,fill=white,minimum width=9.5mm]
\tikzstyle{semilarge dmap conj}=[draw,doubled,shape=NWbox,inner sep=2pt,minimum height=6mm,fill=white,minimum width=9.5mm]
\tikzstyle{large dmap}=[draw,doubled,shape=NEbox,inner sep=2pt,minimum height=6mm,fill=white,minimum width=12mm]
\tikzstyle{large dmap conj}=[draw,doubled,shape=NWbox,inner sep=2pt,minimum height=6mm,fill=white,minimum width=12mm]
\tikzstyle{large dmap trans}=[draw,doubled,shape=SWbox,inner sep=2pt,minimum height=6mm,fill=white,minimum width=12mm]
\tikzstyle{large dmap adj}=[draw,doubled,shape=SEbox,inner sep=2pt,minimum height=6mm,fill=white,minimum width=12mm]
\tikzstyle{large dmap dag}=[draw,doubled,shape=SEbox,inner sep=2pt,minimum height=6mm,fill=white,minimum width=12mm]
\tikzstyle{very large dmap}=[draw,doubled,shape=NEbox,inner sep=2pt,minimum height=6mm,fill=white,minimum width=19.5mm]

\tikzstyle{muxbox}=[draw,shape=rectangle,minimum height=3mm,minimum width=3mm,fill=white]
\tikzstyle{dmuxbox}=[muxbox,doubled]

\tikzstyle{box}=[draw,shape=rectangle,inner sep=2pt,minimum height=6mm,minimum width=6mm,fill=white]
\tikzstyle{dbox}=[draw,doubled,shape=rectangle,inner sep=2pt,minimum height=6mm,minimum width=6mm,fill=white]
\tikzstyle{dmap}=[draw,doubled,shape=NEbox,inner sep=2pt,minimum height=6mm,fill=white]
\tikzstyle{dmapdag}=[draw,doubled,shape=SEbox,inner sep=2pt,minimum height=6mm,fill=white]
\tikzstyle{dmapadj}=[draw,doubled,shape=SEbox,inner sep=2pt,minimum height=6mm,fill=white]
\tikzstyle{dmaptrans}=[draw,doubled,shape=SWbox,inner sep=2pt,minimum height=6mm,fill=white]
\tikzstyle{dmapconj}=[draw,doubled,shape=NWbox,inner sep=2pt,minimum height=6mm,fill=white]

\tikzstyle{ddmap}=[draw,doubled,dashed,shape=NEbox,inner sep=2pt,minimum height=6mm,fill=white]
\tikzstyle{ddmapdag}=[draw,doubled,dashed,shape=SEbox,inner sep=2pt,minimum height=6mm,fill=white]
\tikzstyle{ddmapadj}=[draw,doubled,dashed,shape=SEbox,inner sep=2pt,minimum height=6mm,fill=white]
\tikzstyle{ddmaptrans}=[draw,doubled,dashed,shape=SWbox,inner sep=2pt,minimum height=6mm,fill=white]
\tikzstyle{ddmapconj}=[draw,doubled,dashed,shape=NWbox,inner sep=2pt,minimum height=6mm,fill=white]

\boxshape{sNEbox}{0pt}{3pt}
\boxshape{sSEbox}{0pt}{-3pt}
\boxshape{sNWbox}{3pt}{0pt}
\boxshape{sSWbox}{-3pt}{0pt}
\tikzstyle{smap}=[draw,shape=sNEbox,fill=white]
\tikzstyle{smapdag}=[draw,shape=sSEbox,fill=white]
\tikzstyle{smapadj}=[draw,shape=sSEbox,fill=white]
\tikzstyle{smaptrans}=[draw,shape=sSWbox,fill=white]
\tikzstyle{smapconj}=[draw,shape=sNWbox,fill=white]

\tikzstyle{dsmap}=[draw,dashed,shape=sNEbox,fill=white]
\tikzstyle{dsmapdag}=[draw,dashed,shape=sSEbox,fill=white]
\tikzstyle{dsmaptrans}=[draw,dashed,shape=sSWbox,fill=white]
\tikzstyle{dsmapconj}=[draw,dashed,shape=sNWbox,fill=white]

\boxshape{mNEbox}{0pt}{10pt}
\boxshape{mSEbox}{0pt}{-10pt}
\boxshape{mNWbox}{10pt}{0pt}
\boxshape{mSWbox}{-10pt}{0pt}
\tikzstyle{mmap}=[draw,shape=mNEbox]
\tikzstyle{mmapdag}=[draw,shape=mSEbox]
\tikzstyle{mmaptrans}=[draw,shape=mSWbox]
\tikzstyle{mmapconj}=[draw,shape=mNWbox]

\tikzstyle{mmapgray}=[draw,fill=gray!40!white,shape=mNEbox]
\tikzstyle{smapgray}=[draw,fill=gray!40!white,shape=sNEbox]

\makeatletter

\pgfdeclareshape{cornerpoint}{
\inheritsavedanchors[from=rectangle] 
\inheritanchorborder[from=rectangle]
\inheritanchor[from=rectangle]{center}
\inheritanchor[from=rectangle]{north}
\inheritanchor[from=rectangle]{south}
\inheritanchor[from=rectangle]{west}
\inheritanchor[from=rectangle]{east}
\backgroundpath{
\southwest \pgf@xa=\pgf@x \pgf@ya=\pgf@y
\northeast \pgf@xb=\pgf@x \pgf@yb=\pgf@y

\pgfmathsetmacro{\pgf@shorten@left}{\pgfkeysvalueof{/tikz/shorten left}}
\pgfmathsetmacro{\pgf@shorten@right}{\pgfkeysvalueof{/tikz/shorten right}}

\pgfpathmoveto{\pgfpoint{0.5 * (\pgf@xa + \pgf@xb)}{\pgf@ya - 5pt}}
\pgfpathlineto{\pgfpoint{\pgf@xa - 8pt + \pgf@shorten@left}{\pgf@yb - 1.5 * \pgf@shorten@left}}
\pgfpathlineto{\pgfpoint{\pgf@xa - 8pt + \pgf@shorten@left}{\pgf@yb}}
\pgfpathlineto{\pgfpoint{\pgf@xb + 8pt - \pgf@shorten@right}{\pgf@yb}}
\pgfpathlineto{\pgfpoint{\pgf@xb + 8pt - \pgf@shorten@right}{\pgf@yb - 1.5 * \pgf@shorten@right}}
\pgfpathclose
}
}

\pgfdeclareshape{cornercopoint}{
\inheritsavedanchors[from=rectangle] 
\inheritanchorborder[from=rectangle]
\inheritanchor[from=rectangle]{center}
\inheritanchor[from=rectangle]{north}
\inheritanchor[from=rectangle]{south}
\inheritanchor[from=rectangle]{west}
\inheritanchor[from=rectangle]{east}
\backgroundpath{
\southwest \pgf@xa=\pgf@x \pgf@ya=\pgf@y
\northeast \pgf@xb=\pgf@x \pgf@yb=\pgf@y

\pgfmathsetmacro{\pgf@shorten@left}{\pgfkeysvalueof{/tikz/shorten left}}
\pgfmathsetmacro{\pgf@shorten@right}{\pgfkeysvalueof{/tikz/shorten right}}

\pgfpathmoveto{\pgfpoint{0.5 * (\pgf@xa + \pgf@xb)}{\pgf@yb + 5pt}}
\pgfpathlineto{\pgfpoint{\pgf@xa - 8pt + \pgf@shorten@left}{\pgf@ya + 1.5 * \pgf@shorten@left}}
\pgfpathlineto{\pgfpoint{\pgf@xa - 8pt + \pgf@shorten@left}{\pgf@ya}}
\pgfpathlineto{\pgfpoint{\pgf@xb + 8pt - \pgf@shorten@right}{\pgf@ya}}
\pgfpathlineto{\pgfpoint{\pgf@xb + 8pt - \pgf@shorten@right}{\pgf@ya + 1.5 * \pgf@shorten@right}}
\pgfpathclose
}
}

\makeatother

\pgfkeyssetvalue{/tikz/shorten left}{0pt}
\pgfkeyssetvalue{/tikz/shorten right}{0pt}

\tikzstyle{kpoint common}=[draw,fill=white,inner sep=1pt,minimum height=4mm]
\tikzstyle{kpoint sc}=[shape=cornerpoint,kpoint common]
\tikzstyle{kpoint adjoint sc}=[shape=cornercopoint,kpoint common]
\tikzstyle{kpoint}=[shape=cornerpoint,shorten left=5pt,kpoint common]
\tikzstyle{kpoint adjoint}=[shape=cornercopoint,shorten left=5pt,kpoint common]
\tikzstyle{kpoint conjugate}=[shape=cornerpoint,shorten right=5pt,kpoint common]
\tikzstyle{kpoint transpose}=[shape=cornercopoint,shorten right=5pt,kpoint common]
\tikzstyle{kpoint symm}=[shape=cornerpoint,shorten left=5pt,shorten right=5pt,kpoint common]

\tikzstyle{wide kpoint sc}=[shape=cornerpoint,kpoint common, minimum width=1 cm]
\tikzstyle{wide kpointdag sc}=[shape=cornercopoint,kpoint common, minimum width=1 cm]

\tikzstyle{black kpoint}=[shape=cornerpoint,shorten left=5pt,kpoint common,fill=black,font=\color{white}]

\tikzstyle{black kpoint sm}=[shape=cornerpoint,shorten left=5pt,kpoint common,fill=black,font=\color{white},scale=0.75]

\tikzstyle{black kpoint adjoint}=[shape=cornercopoint,shorten left=5pt,kpoint common,fill=black,font=\color{white}]
\tikzstyle{black kpointadj}=[shape=cornercopoint,shorten left=5pt,kpoint common,fill=black,font=\color{white}]

\tikzstyle{black kpointadj sm}=[shape=cornercopoint,shorten left=5pt,kpoint common,fill=black,font=\color{white},scale=0.75]

\tikzstyle{black dkpoint}=[shape=cornerpoint,shorten left=5pt,kpoint common,fill=black, doubled,font=\color{white}]
\tikzstyle{black dkpoint adjoint}=[shape=cornercopoint,shorten left=5pt,kpoint common,fill=black, doubled,font=\color{white}]
\tikzstyle{black dkpointadj}=[shape=cornercopoint,shorten left=5pt,kpoint common,fill=black, doubled,font=\color{white}]

\tikzstyle{black dkpoint sm}=[shape=cornerpoint,shorten left=5pt,kpoint common,fill=black, doubled,font=\color{white},scale=0.75]
\tikzstyle{black dkpointadj sm}=[shape=cornercopoint,shorten left=5pt,kpoint common,fill=black, doubled,font=\color{white},scale=0.75]

\tikzstyle{kpointdag}=[kpoint adjoint]
\tikzstyle{kpointadj}=[kpoint adjoint]
\tikzstyle{kpointconj}=[kpoint conjugate]
\tikzstyle{kpointtrans}=[kpoint transpose]

\tikzstyle{big kpoint}=[kpoint, minimum width=1.2 cm, minimum height=8mm, inner sep=4pt, text depth=3mm]

\tikzstyle{wide kpoint}=[kpoint, minimum width=1 cm, inner sep=2pt]
\tikzstyle{wide kpointdag}=[kpointdag, minimum width=1 cm, inner sep=2pt]
\tikzstyle{wide kpointconj}=[kpointconj, minimum width=1 cm, inner sep=2pt]
\tikzstyle{wide kpointtrans}=[kpointtrans, minimum width=1 cm, inner sep=2pt]

\tikzstyle{wider kpoint}=[kpoint, minimum width=1.25 cm, inner sep=2pt]
\tikzstyle{wider kpointdag}=[kpointdag, minimum width=1.25 cm, inner sep=2pt]
\tikzstyle{wider kpointconj}=[kpointconj, minimum width=1.25 cm, inner sep=2pt]
\tikzstyle{wider kpointtrans}=[kpointtrans, minimum width=1.25 cm, inner sep=2pt]

\tikzstyle{gray kpoint}=[kpoint,fill=gray!50!white]
\tikzstyle{gray kpointdag}=[kpointdag,fill=gray!50!white]
\tikzstyle{gray kpointadj}=[kpointadj,fill=gray!50!white]
\tikzstyle{gray kpointconj}=[kpointconj,fill=gray!50!white]
\tikzstyle{gray kpointtrans}=[kpointtrans,fill=gray!50!white]

\tikzstyle{gray dkpoint}=[kpoint,fill=gray!50!white,doubled]
\tikzstyle{gray dkpointdag}=[kpointdag,fill=gray!50!white,doubled]
\tikzstyle{gray dkpointadj}=[kpointadj,fill=gray!50!white,doubled]
\tikzstyle{gray dkpointconj}=[kpointconj,fill=gray!50!white,doubled]
\tikzstyle{gray dkpointtrans}=[kpointtrans,fill=gray!50!white,doubled]

\tikzstyle{white label}=[draw,fill=white,rectangle,inner sep=0.7 mm]
\tikzstyle{gray label}=[draw,fill=gray!50!white,rectangle,inner sep=0.7 mm]
\tikzstyle{black label}=[draw,fill=black,rectangle,inner sep=0.7 mm]

\tikzstyle{dkpoint}=[kpoint,doubled]
\tikzstyle{wide dkpoint}=[wide kpoint,doubled]
\tikzstyle{dkpointdag}=[kpoint adjoint,doubled]
\tikzstyle{wide dkpointdag}=[wide kpointdag,doubled]
\tikzstyle{dkcopoint}=[kpoint adjoint,doubled]
\tikzstyle{dkpointadj}=[kpoint adjoint,doubled]
\tikzstyle{dkpointconj}=[kpoint conjugate,doubled]
\tikzstyle{dkpointtrans}=[kpoint transpose,doubled]

\tikzstyle{kscalar}=[kpoint common, shape=EBox, inner xsep=-1pt, inner ysep=3pt,font=\small]
\tikzstyle{kscalarconj}=[kpoint common, shape=WBox, inner xsep=-1pt, inner ysep=3pt,font=\small]

\tikzstyle{spekpoint}=[kpoint sc,minimum height=5mm,inner sep=3pt]
\tikzstyle{spekcopoint}=[kpoint adjoint sc,minimum height=5mm,inner sep=3pt]

\tikzstyle{dspekpoint}=[spekpoint,doubled]
\tikzstyle{dspekcopoint}=[spekcopoint,doubled]

 \tikzstyle{upground}=[circuit ee IEC,thick,ground,rotate=90,scale=2.5]
 \tikzstyle{downground}=[circuit ee IEC,thick,ground,rotate=-90,scale=2.5]
 \tikzstyle{bigground}=[regular polygon,regular polygon sides=3,draw=gray,scale=0.50,inner sep=-0.5pt,minimum width=10mm,fill=gray]

\tikzstyle{arrs}=[-latex,font=\small,auto]
\tikzstyle{arrow plain}=[arrs]
\tikzstyle{arrow dashed}=[dashed,arrs]
\tikzstyle{arrow bold}=[very thick,arrs]
\tikzstyle{arrow hide}=[draw=white!0,-]
\tikzstyle{arrow reverse}=[latex-]
\tikzstyle{cdnode}=[]

%% file: stylesnew5.tikzstyles
\tikzstyle{discarding}=[fill=white, draw=black, shape=circle, style=upground]
\tikzstyle{smalldiscarding}=[fill=white, draw=black, style=upground, scale=0.5]
\tikzstyle{backdiscard}=[fill=white, draw=black, shape=circle, style=downground, scale=0.5]
\tikzstyle{smallbackdiscard}=[fill=white, draw=black, shape=circle, style=downground, scale=0.5]
\tikzstyle{state}=[fill=white, draw=black, style=triang, tikzit shape=rectangle]
\tikzstyle{kstate}=[fill=white, draw=black, style=kpoint, tikzit shape=rectangle]
\tikzstyle{kstateconj}=[fill=white, draw=black, style=kpoint conjugate, tikzit shape=rectangle]
\tikzstyle{kstateBIG}=[fill=white, draw=black, style=big kpoint, tikzit shape=rectangle]
\tikzstyle{effect}=[fill=white, draw=black, style=triangdag]
\tikzstyle{keffect}=[fill=white, draw=black, style=kpoint adjoint]
\tikzstyle{keffectconj}=[fill=white, draw=black, style=kpoint transpose]
\tikzstyle{morphdag}=[style=mapdag]
\tikzstyle{morph}=[style=hadamard]
\tikzstyle{WIDEmorph}=[style=hadamard, minimum width=14mm]
\tikzstyle{morphtrans}=[style=maptrans]
\tikzstyle{morphconj}=[style=mapconj]
\tikzstyle{CPMmorph}=[style=dmap]
\tikzstyle{CPMmorphconj}=[style=dmapconj]
\tikzstyle{CPMmorphdag}=[style=dmapdag]
\tikzstyle{CPMmorphtrans}=[style=dmaptrans]
\tikzstyle{CPMstate}=[fill=white, draw=black, style=triang, doubled]
\tikzstyle{CPMstateBIG}=[fill=white, draw=black, style={triang_lesssep}, doubled]
\tikzstyle{CPMkstate}=[fill=white, draw=black, style=kpoint, tikzit shape=rectangle, doubled]
\tikzstyle{CPMkstateconj}=[fill=white, draw=black, style=kpoint conjugate, tikzit shape=rectangle, doubled]
\tikzstyle{CPMkstateBIG}=[fill=white, draw=black, style=big kpoint, tikzit shape=rectangle, doubled]
\tikzstyle{CPMkeffect}=[fill=white, draw=black, style=kpoint adjoint, doubled]
\tikzstyle{CPMkeffectconj}=[fill=white, draw=black, style=kpoint transpose, doubled]
\tikzstyle{UHfB}=[fill=white, draw=black, style=triangdag, doubled, inner sep=-2pt]
\tikzstyle{leak}=[style=tinypoint, regular polygon rotate=-90]
\tikzstyle{leakfill}=[style=tinypoint, regular polygon rotate=-90, fill=black]
\tikzstyle{Z}=[style=dot, fill=green]
\tikzstyle{X}=[style=dot, fill=red]
\tikzstyle{black_dot}=[style=dot, fill=black]
\tikzstyle{white_dot}=[style=dot, fill=white]
\tikzstyle{qblack_dot}=[style=ddot, fill=black]
\tikzstyle{qwhite_dot}=[style=ddot, fill=white]
\tikzstyle{whitephase}=[style=wphase dot, fill=white]
\tikzstyle{qredphase}=[style=phase dot, fill=red]
\tikzstyle{qgreenphase}=[style=phase dot, fill=green]
\tikzstyle{had}=[style=hadamard, doubled]
\tikzstyle{box}=[style=hadamard]
\tikzstyle{classhad}=[style=hadamard]
\tikzstyle{antipode}=[style=anti]

\tikzstyle{dottededge}=[-, dotted]
\tikzstyle{double edge}=[-, style=doubled, draw=black, tikzit draw={rgb,255: red,234; green,209; blue,255}]
\tikzstyle{new edge style 0}=[<-]
\tikzstyle{new edge style 1}=[-, draw={rgb,255: red,234; green,209; blue,255}, fill={rgb,255: red,234; green,209; blue,255}]
\tikzstyle{new edge style 1b}=[-, draw={rgb,255: red,208; green,218; blue,216}, fill={rgb,255: red,208; green,218; blue,216}]
\tikzstyle{new edge style 2}=[-, draw={rgb,255: red,14; green,188; blue,83}]
\tikzstyle{new edge style 3}=[<-, draw={rgb,255: red,234; green,209; blue,255}]
\tikzstyle{new edge style 4}=[<-, draw={rgb,255: red,0; green,128; blue,128}]
\tikzstyle{new edge style 5}=[-, draw={rgb,255: red,214; green,110; blue,62}]
\tikzstyle{new edge style 6}=[-, draw={rgb,255: red,174; green,20; blue,174}]

%% file: figs/category_1.tikz
\begin{tikzpicture}[scale=0.6, {every node/.style}={scale=0.6}]
	\begin{pgfonlayer}{nodelayer}
		\node [style=none] (208) at (0, 0.25) {};
		\node [style=none] (214) at (0, -0.25) {};
		\node [style=none] (215) at (0.25, -0.5) {$A$};
	\end{pgfonlayer}
	\begin{pgfonlayer}{edgelayer}
		\draw (208.center) to (214.center);
	\end{pgfonlayer}
\end{tikzpicture}

%% file: figs/discard.tikz
\begin{tikzpicture}[scale=0.6, {every node/.style}={scale=0.6}]
	\begin{pgfonlayer}{nodelayer}
		\node [style=smalldiscarding] (0) at (0, 0.5) {};
		\node [style=none] (1) at (0, -0.25) {};
	\end{pgfonlayer}
	\begin{pgfonlayer}{edgelayer}
		\draw (1.center) to (0);
	\end{pgfonlayer}
\end{tikzpicture}

%% file: figs/compact_2.tikz
\begin{tikzpicture}[scale=0.6, {every node/.style}={scale=0.6}]
	\begin{pgfonlayer}{nodelayer}
		\node [style=none] (273) at (0, 0.75) {};
		\node [style=none] (274) at (0, -0.75) {};
	\end{pgfonlayer}
	\begin{pgfonlayer}{edgelayer}
		\draw (274.center) to (273.center);
	\end{pgfonlayer}
\end{tikzpicture}

%% file: figs/wire.tikz
\begin{tikzpicture}[scale=0.6, {every node/.style}={scale=0.6}]
	\begin{pgfonlayer}{nodelayer}
		\node [style=none] (29) at (0, 0.5) {};
		\node [style=none] (35) at (0, -0.5) {};
	\end{pgfonlayer}
	\begin{pgfonlayer}{edgelayer}
		\draw (35.center) to (29.center);
	\end{pgfonlayer}
\end{tikzpicture}